\documentclass[11pt]{amsart}
\usepackage[left=1.5in, right=1.5in]{geometry}
\usepackage[latin1]{inputenc}
\usepackage{amsfonts}
\usepackage[toc,page,title,titletoc,header]{appendix}
\usepackage{graphicx,psfrag,epsfig,multirow,caption,subcaption}
\usepackage{amssymb,amsmath,amscd,amsthm,amssymb,verbatim,setspace}
\numberwithin{equation}{section}
\usepackage{mathrsfs,wrapfig}
\usepackage{indentfirst}
\usepackage{extarrows}
\usepackage[colorlinks=true]{hyperref}

\newtheorem{Thm}{Theorem}[section]
\newtheorem{Def}[Thm]{Definition}
\newtheorem{Lm}[Thm]{Lemma}
\newtheorem{Prop}[Thm]{Proposition}

\def\bmt{\left[\begin{array}}
\def\emt{\end{array}\right]}

\def\Z{\mathbb Z}
\def\T{\mathbb T}

\def\R{\mathbb R}
\def\eps{\varepsilon}
\def\al{\alpha}

\def\dt{\delta}

\def\lb{\lambda}

\def\cM{\mathcal M}

\title[KAM and blackholes]{KAM and geodesic dynamics of blackholes}
\author{Jinxin Xue}
\date\today
\begin{document}
\address{Jingzhai 310, Department of mathematics and Yau Mathematical Sciences Center, Tsinghua University, Beijing,100084 }
\email{jxue@tsinghua.edu.cn}%
\maketitle
\begin{abstract}

In this paper we apply KAM theory and the Aubry-Mather theory for twist maps to the study of bound geodesic dynamics of a perturbed blackhole background. The general theories apply mainly to two observable phenomena: the photon shell (unstable bound spherical orbits) and the quasi-periodic oscillations. We discover there is a gap structure in the photon shell that can be used to reveal information of the perturbation.


\end{abstract}

\tableofcontents
\renewcommand\contentsname{Index}

\section{Introduction}
This paper is a companion paper of \cite{X}. The main purpose is to study the stability of the geodesic dynamics of a blackhole background under perturbations. 

The important metrics such as Schwarzschild and Kerr are both integrable, each with four independent conserved quantities. If we consider the bounded motions in an integrable system, it is known from Liouville-Arnold theorem that  the dynamics can be treated as toral linear flows.    When an integrable system is perturbed slightly, the integrability is generically broken and chaotic motions occur. The remarkable discovery of the Kolmogorov-Arnold-Moser theorem shows that most motions are still toral linear flows after a coordinate change. In reality, the blackholes modeled on either Schwarzschild and Kerr always undergo some perturbations. We try to understand the dynamics of a particle, massless or mass, moving in a perturbed Schwarzschild or Kerr background, using the theory of Hamiltonian dynamical systems developed for nearly integrable systems.

  The mathematical model that we use to study the geodesic dynamics of the blackholes is the following map called integrable twist map:
   \begin{equation}\label{EqTwist0}\phi_0: \ \T^1\times [0,1]\to \T^1\times [0,1],\quad (\theta, I)\mapsto (\theta+\nu(I), I)\end{equation}
 where $\nu$ is assumed to be smooth and strictly monotone in $I$. The dynamics of this example is as follows.  Each $I$-circle is invariant and the dynamics on it is a rotation by $\nu(I)$. We call this rotation rate $\nu(I)$ the rotation number. When $\nu(I)$ is rational, then each orbit on the $I$-circle is periodic, and when $\nu(I)$ is irrational, then each orbit is dense on the $I$-circle. When $\phi_0$ is slightly perturbed within the class of smooth symplectic maps, we have the following picture for the dynamics. First,  Moser's version of KAM theory implies that each $I$-circle with $\nu(I)$ a Diophantine number is perturbed to a nearby smooth loop, called invariant circle, that remains invariant under the perturbed map and the dynamics on it can be conjugate to the original unperturbed rotation, provided the perturbation is sufficiently small in a smooth enough topology. For a fixed small perturbation of size $\eps$, the set of remaining invariant circles has rotation numbers lying in a Cantor set (a nowhere dense closed set) of measure $1-O(\sqrt\eps)$. Between two nearby invariant circles (not necessarily KAM curve), there is a gap region called Birkhoff instability region where the dynamics is very chaotic. There is a Aubry-Mather theory developed for general twist maps which gives the existence of some special orbits in the Birkhoff instability region, such as periodic orbits, heteroclinic orbits corresponding to rational rotation numbers and Cantor set like orbit corresponding to irrational rotation numbers. So KAM theory allows us to find a large measure set of the phase space where the motion is regular (quasiperiodic), while Aubry-Mather theory allows use to find some special orbits in the gaps.  We give an outline of the two theories in Section \ref{SSTwist} and more details in Appendix \ref{AppAM}.

We first locate the part of phase space with bounded motions. For massless particles moving on null geodesics in Schwarzschild or Kerr background, these bounded motions are called bound photon orbits, or fundamental photon orbits in literature. Each such orbit is moving on a sphere with a fixed radius and is unstable under radial perturbations. This is an observable feature of the blackhole, which lies on the edges of the blackhole shadows  \cite{M87,N}. 
 If we also consider the motion of massive particles on timelike geodesics in Schwarzschild or Kerr background, again we have unstable bound spherical orbits, similar to the lightlike case.
 For simplicity, we use the term {\it photon shell} to call the set of bound spherical orbits that are unstable under radial perturbations, for both Schwarzschild and Kerr, and both lightlike or timelike geodesics. The photon shell when considered in the phase space corresponds to a geometric object called normally hyperbolic invariant manifold in dynamical systems, which is known to be robust under perturbations. This means that an orbit on the photon shell remains a bound orbit in a perturbed Schwarzschild or Kerr metric, though it may become chaotic. When a background is a perturbation of Schwarzschild or Kerr, we still use the terminology {\it photon shell} to call the union of bound orbits that are unstable under radial perturbations. Note that in this case, the radial component of each orbit may oscillate slightly.

When studying Arnold diffusion, people often need some information on the normally hyperbolic invariant manifold, which in our case is the photon shell (c.f. \cite{X}).  We show in this paper that both KAM and the twist map theory apply to the study of the dynamics on the photon shell in a perturbed Kerr spacetime, but neither applies to perturbed Schwarzschild a priori without further information of the perturbation. In particular, KAM implies in the Kerr case the existence of a gap structure (many Birkhoff instability regions) on the photon shell exhibiting coexistence of regular and chaotic motions, which we believe is observable and can reveal information of the perturbation.


Furthermore, in the massive case, in addition to the photon shell, there are a lot more bound timelike geodesics. It can be observed that outside a blackhole that is a slowly rotating accretion disk emitting X-rays exhibiting some quasi-periodic oscillations (QPO) in its frequency (c.f. \cite{DKN, IM, JDW} etc). In the phase space this is a region around a stable circular orbit. We show both KAM and twist map theory apply in this region for both Schwarzchild and Kerr cases to yield lots of quasiperiodic orbits as well as other special orbits.

Both Schwarzschild and Kerr metrics are written in coordinates $(\tau,r,\theta,\varphi)$ where $\tau$ is the coordinate time, $r$ is the polar radius, $\theta$ is the latitudinal angle and $\varphi$ is the azimuthal angle. We will use the following definition throughout the paper.
\begin{Def}
A metric is called stationary if it is independent of the coordinate time $\tau$ and axisymmetric if it is independent of the azimuth angle $\varphi$.
\end{Def}

\subsection{The photon shell dynamics}

We consider perturbations stationary perturbation so that the particle's energy $E$, that is dual to $\tau,$ remains to be a constant of motion, so we fix a value $E$ and restrict to the photon shell. The resulting system has two degrees of freedom with coordinates $(\theta,p_\theta, \varphi, L_z)$ and we only need to study the latitudinal and azimuthal motions, where $p_\theta$ and $L_z=p_\varphi$ are generalized momentum variables due to $\theta$ and $\varphi$ respectively, and $L_z$ in particular is a constant of motion having the physical meaning of the $z$-component of the angular momentum. This subsystem describes a particle's motion on the photon shell that has a vertical oscillation within a spherical band symmetric around the equator when orbiting around the center of the blackhole. 

Our general strategy is to introduce action-angle coordinates $(\al_\theta,J_\theta,\al_\varphi, L_z)\in \T\times \R_+\times \T\times \R$ for the $\theta$- and $\varphi$-motions. In these coordinates the unperturbed Schwarzschild or Kerr Hamiltonian has the form $H(J_\theta, L_z)$, so the dynamics leaves $J_\theta,L_z$ invariant and the angles $(\al_\theta,\al_\varphi)$ move on the torus $\T^2$ linearly as $(\al_\theta,\al_\varphi)\mapsto (\al_\theta,\al_\varphi)+t (\partial_{J_\theta}H, \partial_{L_z}H)\ \mathrm{mod}\ \Z^2,\ t\in \R$. For this system, we further fix a Hamiltonian level set $H=-\frac12\mu^2$ with $\mu^2=0$ in the massless case and $\mu^2=1$ in the massive case, and look at the return map to the section $\{\al_\theta=0\}$ (or to say $\{\theta=\pi/2\}$). Then the map has the form of \eqref{EqTwist0} with $\nu:=\partial_{L_z}H/\partial_{J_\theta}H$ parametrized by $L_z$. To apply KAM theory and twist map theory, we have to find a parameter range where the ratio has no critical point. The main difficulty is that the Hamiltonian cannot be written explicitly as the action variables $(J_\theta,L_z)$.

The Schwarzschild case is relatively simple, since each photon spherical orbit moves on a plane passing through the origin so is periodic and all the photon spherical orbits have the same period. In this case the ratio $\nu$ is constantly 1, so neither KAM nor twist map theory applies without further knowledge of the perturbations.

It turns out that both theories apply to Kerr. For Kerr, we can reparametrize the ratio $\nu$ by the radius $r$ of the photon shell orbits and denote by $\mathcal R$ the domain of definition for $r$ (c.f. Definition \ref{DefAdmE}). So the photon shell of the unperturbed Kerr is foliated by $S_r$, that is a spherical strip symmetric around the equator consisting of all spherical orbits of radius $r\in \mathcal R$. 

\begin{Thm}\label{ThmPhotonKerr0}
Consider null geodesics in a  Kerr spacetime with mass $M$ and angular momentum $a$. Then there exists an open set of parameter $a/M$ such that the frequency ratio $\nu=\omega_\varphi/\omega_\theta$ as a function of $r\in \mathcal R$ is smooth, strictly monotone and has a discontinuity point at some $r_\star\in \mathcal R$ with $\nu(r_{\star-})-\nu(r_{\star+})=2$. Orbits on $S_r$ with $r<r_\star$ are all prograde and that with $r>r_\star$ are retrograde. Denote by $U$ a small neighborhood of $r_\star$.

Let $ \eps h_{\mu\nu}dx^\mu dx^\nu$ be a stationary perturbation of the above Kerr metric. Then if $\eps$ sufficiently small, we have the following for the dynamics of null geodesics in the perturbed system
\begin{enumerate}
\item There is a closed subset $\mathcal C$ of $\mathcal R\setminus U$ such that for each $r\in \mathcal C$, in the perturbed system, there is a set $\tilde S_r$ that is a small deformation of $S_r$, such that $\tilde S_r$ consists of periodic or quasiperiodic orbits.
\item Each interval in $\mathcal R\setminus (U\cup\mathcal C)$ corresponds to a Birkhoff instability region where the spheres $S_r$ are broken (c.f. Definition \ref{DefBroken}) and twist map theory $($c.f. Theorem \ref{ThmAM} $(2.b), (3.b)$ and Theorem \ref{ThmMather}$)$ applies.
\end{enumerate}
\end{Thm}

 \begin{figure}
\centering
\includegraphics[width=0.4\textwidth]{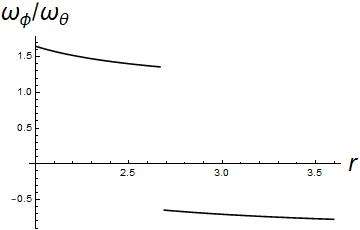}
\caption{The frequency ratio $\omega_\varphi/\omega_\theta$ with parameters $a/M=0.8, \mu^2=0$}
\label{FigPhotonK0}
\end{figure}

In particular, the gaps (instability regions) often occur around rational frequency ratios of the form $p/q$, due to the existence of Fourier modes $\exp(\sqrt{-1}k(p\al_\theta+q\al_\varphi)),$ $k,p,q\in \Z,$ in the perturbation. In general, the gap is large when the denominator of the rational number is small. We expect that this is an observable phenomenon, in some sense an analogue of the Kirkwood gap in the asteroid belt in the solar system, and can be used to conclude the information of the perturbation.

In Figure \ref{FigPhotonK0}, we see that there is a discontinuity in the frequency ratio $\nu$, which occurs at a point where $L_z$ changes sign and the motion of the particle becomes from prograde to retrograde. This phenomenon was observed in \cite{Wi,T}. We give a elementary mathematical explanation in Section \ref{SSSDisc}. In the proof, we deal with a particularly chosen parameter value $a/M=0.8$, but it is clear that our formulas are applicable to any other parameter values of interest. 

For the geodesic motion of massive particles, we have a similar theorem (c.f. Theorem \ref{ThmPhotonKerr}). The main difference of the massive case and the massless case is that the dynamics on the photon shell of the massive case depends on the particle's $E$, while the massless case does not. So in the massive case, we need to fix an $E$ ahead of time, then there is the similar gap structure for the photon orbits of all massive particles with the same energy $E$.

  Note that our result can be combined with the result of Johnson, Lupsasca, Strominger et al \cite{J} to give a delicate picture of the fine structure of the neighborhood of the photon shell.  It is called the photon ring the image on the observer's screen of photon's on nearly bound null geodesics. So the photon ring contains the photon shell. It was discovered in \cite{J} that there is a subring structure outside the photon shell. Indeed, if we shoot a light ray very near, a distance from the shadow edge at $\delta r_0$, it will circle $n\sim -\frac{1}{\gamma}\ln |\delta r_0|$ times  before falling into the black hole or escaping to infinity, due to the instability in the direction normal to the photon shell, where $\gamma$ is the Lyapunov exponent. A light ray that completes $n$ half orbits collects $n$ times more photons along its path.
This gives a ring on the screen subdivided into subrings labeled by the number of circulations $n$, exponentially narrower when approaching the photon shell.

 \subsection{Quasi-periodic oscillations (QPO)}
 In the case of massive particles, the geodesic dynamics has lots of bounded motions other than the bound spherical orbits. These bounded motions occur around the stable spherical orbits, thus they undergo both radial oscillations and the latitudinal oscillations when spinning around the blackhole. They are supposed to be responsible for the QPO behavior in some theories.  

We study the dynamics of these bound orbits using similar approaches as before. For simplicity, we consider perturbations that are both stationary and axisymmetric hence $E$ and $L_z$ remain to be constants of motion. Then we can use coordinates $(r,p_r,\theta,p_\theta)$ to study the radial and latitudinal dynamics of the Hamiltonian system. We   introduce action-angle coordinates $(\al_r,J_r,\al_\theta,J_\theta)\in \T\times \R_+\times \T\times \R_+$ so that the Hamiltonian is a function of $J_r$ and $J_\theta$ only and the radial frequency $\partial_{J_r}H$ and vertical frequency $\partial_{J_\theta}H$ are called fundamental frequencies in \cite{S}.
The ratio $\nu=\partial_{J_r}H/\partial_{J_\theta}H$ plays a similar role to the photon shell case as above. Again we can compute $\nu$ to verify that the assumptions for both KAM and twist map theory apply, so we have an analogue of Theorem \ref{ThmPhotonKerr0} in both Schwarzschild and Kerr cases. We refer readers to Theorem \ref{ThmQPOS} and \ref{ThmQPOK} for more details.

If we consider non axisymmetric perturbations, then we should consider a Hamiltonian system of three degrees of freedom with coordinates $(r,p_r,\theta,p_\theta,\varphi,L_z)$. The KAM nondegeneracy condition can be studied similarly. We shall outline how this can be done without working on details. Action-angle coordinates for Kerr has been constructed in \cite{S}. In general, we do not expect that KAM nondegeneracy condition holds globally in the part of phase space with bounded motions for Kerr, considering the isofrequency pairing results (c.f. \cite{WBS} etc).

\noindent {\bf Organization of the paper} The paper is organized as follows. The main body of the paper consists mainly of introductory and conceptual arguments and statements, and we put all the technical ingredients to the appendix.
 In Section \ref{SIntroHam}, we give a brief introduction to nearly integrable Hamiltonian dynamics including Liouville-Arnold theorem, KAM theorem,  and twist map theory.  In Section \ref{SS}, we study the perturbed Schwarzschild dynamics. In Section \ref{SSK}, we study the dynamics of perturbed Kerr. Finally, we have two appendices. In Appendix \ref{SNHIM} we introduce the theorem of normally hyperbolic invariant manifolds, and in Appendix \ref{AppAM}, we introduce the Aubry-Mather theory of twist maps. 

 \section{Introduction to nearly integrable Hamiltonian dynamics}\label{SIntroHam}
 In the Hamiltonian formalism of the classical mechanics, a smooth Hamiltonian function $H$ on a symplectic manifold $(M,\omega)$ is given, and defines a vector field $X_H$ through $\omega(\cdot, X_H)=dH$ which defines a dynamical system by solving the ODE $\dot x=X_H(x)$.
 \subsection{Liouville-Arnold theorem}
The classical Liouville-Arnold theorem gives important characterization of the integrable systems.

\begin{Thm}[Liouville-Arnold] \label{ThmLA}Let $H_1=H:\ M^{2n}\to \R$ be a Hamiltonian and suppose there are $H_2,\ldots,H_n:\ M\to \R$ satisfying
\begin{enumerate}
\item[(a)] $\{H_i,H_j\}\equiv0$, for all $i,j=1,\ldots,n$, where $\{H_i,H_j\}:=\omega(X_{H_i},X_{H_j})$ is the Poisson bracket;
\item[(b)] the level set $M_{\mathbf a}:=\{(q,p)\in M\ |\ H_i(q,p)=a_i,\ i=1,\ldots,n\}$ is compact;
\item[(c)] At each point of $M_{\mathbf a}$, the $n$ vectors $DH_i,\ i=1,\ldots,n$ are linearly independent.
\end{enumerate}
Then
\begin{enumerate}
\item $M_{\mathbf a}$ is diffeomorphic to $\T^n=\R^n/\Z^n$ and is invariant under the Hamiltonian flow of each $H_i$.
\item in a neighborhood $U$ of $M_{\mathbf a}$, there is a symplectic transform $\Phi(q,p)=(\theta,I)$ such that $\Phi(U)=\T^n\times (-\dt,\dt)^n$ for some $\dt>0$.
\item In the new coordinates, each $K_i:=H_i\circ \Phi^{-1}$ is a function of $I$ only so the Hamiltonian equation is
$$\dot \theta=\omega_i(I):=\frac{\partial K_i}{\partial I},\quad \dot I=0. $$
\end{enumerate}
\end{Thm}
We refer readers to Section 50 of \cite{A89} for the proof and more details. The coordinates $(\theta, I)$ are called action-angle coordinates and can be constructed by standard recipe in \cite{A89}. 

 \subsection{The Kolmogorov-Arnold-Moser theorem}\label{SSKAM}
 The Kolmogorov-Arnold-Moser theorem is an important stability result on the dynamics of nearly integrable systems.
\begin{Def}
A vector $\al\in \R^n, \ n\geq 2$ is called \emph{Diophantine} if there exist some $C>0,\tau>n-1$ such that for all $k\in \Z^n\setminus\{0\}$, we have
$$|\langle\al,k\rangle|\geq \frac{C}{|k|^\tau}. $$
 A number $\al\in\R$ is called a Diophantine number, if $(1,\al)\in \R^2$ is a Diophantine vector. In both cases, we denote $\al\in \mathrm{DC}(C,\tau)$.
\end{Def}
For each $\tau>n-1$, the set $\cup_{C>0}\mathrm{DC}(C,\tau)$ has full Lebesgue measure in $\R^n$. For a fixed $C$, the set $\mathrm{DC}(C,\tau)$ is a Cantor set of positive measure.

A version of the KAM theorem states as follows.
\begin{Thm}\label{ThmKAM}
Suppose $\det D^2 h(I)\neq 0$ for $I\in D$. Then for all $C>0,\tau>n-1$ and all $I\in D$ with $\omega(I)\in \mathrm{DC}(C,\tau)$, there exist $\eps_0>0$ and $\ell$ sufficiently large such that when $\eps\|f\|_{C^\ell}<\eps_0$, we have that there is a torus $\mathcal T_I$ that is invariant under the Hamiltonian flow generated by the Hamiltonian $H(\theta,I)=h(I)+\eps f(\theta,I)$. Moreover, there is a symplectic transform $\Phi$ whose restriction to $\mathcal T_I$ satisfies $\Phi(\mathcal T_I)=\T^n$ and the conjugated flow $\Phi\circ\phi^t_H\circ\Phi^{-1}$ is the translation $\theta\mapsto \theta+\omega(I)t$ when restricted to $\T^n$.
\end{Thm}
We refer readers to \cite{P} for more details. When applying KAM theory to an energy level, we may replace the nondegeneracy condition $\det D^2 h(I)\neq 0$ by Arnold's isoenergetic nondegeneracy condition $\det\left[\begin{array}{cc}
D^2h(I)& Dh(I)\\
(Dh(I))^t&0
\end{array}\right]\ne 0$
(c.f. Appendix 8 of \cite{A89} and Chapter 6.3 of \cite{AKN} for various versions of nondegeneracy conditions). 

\subsection{Energetic  reduction}\label{SSERed}
 Let $ H:\ T^*\T^2\to \R$ be a Hamiltonian system of two degrees of freedom of the form $ H(x_1,x_2,y_1,y_2)$ where $(x_1,x_2)\in \T^2$ and $(y_1,y_2)\in \R^2$. Suppose we have $\frac{\partial H}{\partial y_1}\neq 0,$ then we can fix an energy level $E$ and solve the equation $H(x_1,x_2,y_1,y_2)=E$ for $y_1=y_1(x_1,x_2,y_2)$. We can then treat $-y_1$ as a nonautonomous Hamiltonian defined on $T^*\T^1\times \T^1$ where $x_1\in \T^1$ plays the role of new time. By implicit function theorem, we have the Hamiltonian equations \begin{equation}\label{EqEnergeticReduction}
 \begin{cases}
 \frac{dx_2}{dx_1}&=\frac{\partial H}{\partial y_2}\Big/\frac{\partial H}{\partial y_1}=\frac{\partial (-y_1)}{\partial y_2},\\
  \frac{dy_2}{dx_1}&=-\frac{\partial H}{\partial x_2}\Big/\frac{\partial H}{\partial y_1}=-\frac{\partial (-y_1)}{\partial y_2}.
  \end{cases}
 \end{equation} This procedure is called the energetic reduction (Section 45 of \cite{A89}).

 \subsection{Dynamics of twist maps}\label{SSTwist}
 Twist map is defined as follows.
 \begin{Def}\label{DefTwist}
 A diffeomorphism $\phi: \T^1\times [0,1]\to \T^1\times [0,1]$ preserving the boundary is called an area preserving twist map if the following holds:
 \begin{enumerate}
 \item $\phi$ is symplectic, i.e. $\phi^*\omega=\omega,\ \omega=dx\wedge dy;$
 \item for each $x\in \T$, the map $\pi_1 \phi:\ \{x\}\times [0,1]\to \T^1$ is a local homeomorphism, where $\pi_1$ means the projection to the first $(\T^1)$ component.
\end{enumerate}
 \end{Def}
 One example of a twist map is the map \eqref{EqTwist0}. 
 There is a version of KAM theorem developed for twist maps by Moser \cite{M2} that we shall mainly apply later in the paper.
 \begin{Thm}
 Let $\phi_0:\ \T^1\times [0,1]\to \T^1\times [0,1]$ be a twist  map as in \eqref{EqTwist0} and $\phi_\eps:=\phi_0+\eps f:\ \T^1\times [0,1]\to \T^1\times [0,1]$ be a perturbation of $\phi_0$. For all $C>0,\tau>1$, and all $I\in [0,1]$ with $\nu(I)\in \mathrm{DC}(C,\tau)$, there exists $\eps_0>0$ and $\ell$ sufficiently large such that when $\|\eps f\|_{C^\ell}<\eps_0$, we have that there is a homotopically nontrivial circle $\gamma_I$ that is invariant under the map $\phi_\eps$ and the dynamics $\phi_\eps:\ \gamma_{I_*,\eps}\to \gamma_{I_*,\eps}$ can be conjugated to the rotation $\theta\mapsto \theta+\nu(I_*)$ on $\T^1$.
 \end{Thm}
 Another important example of twist maps is the standard map defined as follows: $\Phi:\ \T^2\to \T^2$ via $$(x,y)\mapsto (x+y+k \sin x, y+k \sin x)$$
 where $k$ is a parameter.

 For a general nonlinear twist map (not necessarily nearly integrable), we can introduce a variational principle and define globally minimizing orbits which also has the notion of rotation number playing a similar role as $\nu(I)$ as in \eqref{EqTwist0}. In Aubry-Mather theory, the globally minimizing orbits are classified according to their rotation numbers, denoted by $\mathcal M_\rho$. We give an overview of this theory in Appendix \ref{AppAM}.

 The Aubry-Mather theory gives us the following picture for the dynamics of a twist map. Let $\phi_\eps$ be a smooth twist map and when $\eps=0$, the map $\phi_0$ is of the form \eqref{EqTwist0}. Then for $\eps$ sufficiently small, we have

 \begin{enumerate}
 \item KAM implies that for Diophantine irrational rotation numbers ($\rho \in DC(C,\tau)$ for $C>0$ independent of $\eps$) the set $\cM_\rho$ is a invariant curve restricted on which the map $\phi$ is conjugate to an irrational rotation, i.e. there exists a diffeomorphism $h:\ \cM_\rho\to \T$ such that $h\phi=h+\rho.$
 \item For each irrational rotation number $\rho$, the set $\cM_\rho$ is either a homotopically nontrivial invariant curve, or a Denjoy minimal set.
 \item For rational rotation numbers $\rho$, it may happen that there exists a homotopically nontrivial curve consists of periodic orbits with the same rotation number $\rho\in \mathbb Q$. However, typically, such curve does not exist and $\cM_\rho$ consists of one single periodic orbit with two homoclinic orbits approaching it in both the future and the past.
 \item Whenever there is a region bounded by two neighboring invariant circles, there exist orbits crossing the gap and visiting any two neighborhoods of the boundary circles.
 \item In a neighborhood of an elliptic periodic point of least period $q$ ($D\phi^q$ has no eigenvalues off the unit circle), the above (1),(2),(3),(4) apply also when the twist condition is satisfied. Such a neighborhood is called an elliptic island.
 \end{enumerate}

\section{The Schwarzschild spacetime}\label{SS}
 \subsection{Schwarzschild dynamics}\label{SSMetricS}
 The Schwarzschild metric is as follows
 \begin{equation}
 ds^2=-\left(1-\frac{2M}{r}\right)d\tau^2+\frac{1}{1-\frac{2M}{r}} dr^2+r^2(d\theta^2+\sin^2\theta d\varphi^2)=g_{\mu\nu}dx^\mu dx^\nu
\end{equation}
 with $x^0=\tau$ the time coordinate, $x^1=r$ is the polar radius, and $x^2=\theta\in [0,\pi], x^3=\varphi\in[0,2\pi)$ are the spherical coordinates where $\theta$ is the latitudinal angle and $\varphi$ is the azimuthal angle. Here $M$ is the mass of the blackhole, and $r=2M$ is the event horizon. We are only interested in the dynamics of a particle moving outside of the event horizon. In the following, we denote $\al=\left(1-\frac{2M}{r}\right)$.

 We view the metric as twice of a Lagrangian $L(x,\dot x)=\frac12 g_{\mu\nu}\dot x^\mu \dot x^\nu$ (here $\dot x=\frac{dx}{dt}$ means the derivative with respect to the particle's proper time   $t$) and obtain its conjugate Hamiltonian after a formal Legendre transform as follows
 \begin{equation}\label{EqHamS}
\begin{aligned}
& \begin{cases}
 \frac{d\tau}{dt}&=-\al^{-1}p_\tau, \\
 \frac{dr}{dt}&=\al p_r,\\
 \frac{d\theta}{dt}&= \frac{p_\theta}{r^2},\\
  \frac{d\varphi}{dt}&=\frac{p_\varphi}{r^2\sin^2\theta},\\
\end{cases}\\
  2H(x,p)&=2(\dot x^\mu p_{x^\mu}-L(x,\dot x))=g^{\mu\nu}p_\mu p_\nu\\
  &=-\al^{-1}p_\tau^2+\al p_r^2+r^{-2}\left(p_\theta^2+\frac{1}{\sin^2\theta} p_\varphi^2\right).
\end{aligned}
\end{equation}
The Hamiltonian system has four conserved quantities as $H, p_\tau, p_\varphi, L^2=p_\theta^2+\frac{1}{\sin^2\theta} p_\varphi^2$ with the physical meanings: the total Hamiltonian, the particle energy, the third component of the angular momentum and the square of the total angular momentum respectively. In the following, we also use $L_z$ for $p_\varphi$ and $E$ for $p_\tau$.

We consider a particle moving in the Schwarzschild spacetime along geodesics with invariant mass $\mu$, so along a geodesic, we have
$g_{\mu\nu} \frac{dx^\mu}{dt}\frac{dx^\nu}{dt}=-\mu^2. $
We choose $\mu^2=0$ for massless particles and $\mu^2=1$ for massive particles, and the geodesics in the former case is called a lightlike geodesic and the latter timelike.
\subsubsection{The radial dynamics, critical points and homoclinic orbits}\label{SSSABC}
We next analyze the radial dynamics. Setting $p_\tau=E$, $H=-\frac12\mu^2$ and $\frac{dr}{dt}=\al  p_r$, we obtain \eqref{EqHamS}
 \begin{equation}\label{ThmRadialS}
 \begin{aligned}
\frac{1}{2}E^2&=\frac12\left(\frac{dr}{dt}\right)^2+V(r),\\
2V(r)&=\left(\mu^2+\frac{L^2}{r^2}\right)\al =\mu^2-\frac{2\mu^2 M}{r}+\frac{L^2}{r^2}-\frac{2ML^2}{r^3}.
\end{aligned}
\end{equation}
Thus we may visualize the radial dynamics in the $(r,\frac{dr}{dt})$-plane as a particle moving in the potential well of $V$. The critical points of $V$ corresponds to orbits with $r=\mathrm{const}$. Setting
$\frac{dV}{dr}=0$ we obtain
$$\mu^2 M r^2-L^2 r+3ML^2=0$$
with solutions
\begin{equation}\label{Eqrpm}r_\pm=\frac{L^2\pm \sqrt{L^4-12M^2L^2}}{2M} \end{equation}
for $\mu^2=1,$ and $r=3M$ for $\mu=0$.

The case of $\mu^2=0$ differs drastically from the case of $\mu^2=1$. Indeed, when $\mu^2=0$, we have $V(r)=\frac{L^2}{r^2}-\frac{2ML^2}{r^3}\to -\infty$ as $r\to 0$, $V(\infty)=0$, and the point $r=3M$ is the global maximum point and the unique critical point of $V$. This means that the bound spherical orbits are the only bound orbits in the lightlike case.

\begin{figure}
\begin{minipage}[t]{0.5\linewidth}
\centering
\includegraphics[width=2.2in]{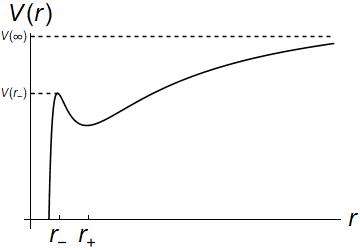}
\subcaption{$\mu^2=1$, $L^2>L^2_{isco}$ but not large}
\label{fig:side:a}
\end{minipage}%
\begin{minipage}[t]{0.5\linewidth}
\centering
\includegraphics[width=2.2in]{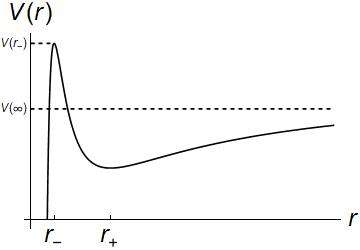}
\subcaption{$\mu^2=1$, large $L^2$}
\label{fig:side:b}
\end{minipage}
\caption{The Schwarzschild effective potential}
\label{FigV}
\end{figure}

We consider next the case $\mu^2=1$. For $L$ satisfying $L^2>12M^2$, there are two solutions $r_+\neq r_-$ by \eqref{Eqrpm}. The root $r_-$ is a local maximum of $V$ and $r_+$ is a local minimum. All orbits with initial value $r<r_-$ plunge into the blackhole.
When $L^2=12M^2$, the two critical points of $V$ coincide. The value $12M^2$ is called $L_{isco}^2$ where \emph{isco} stands for the \emph{innermost stable circular orbit}. For all values $L^2<L^2_{isco}$, all orbits plunge into the central blackhole.  The set $\{r=r_-\}$ is called the \emph{photon sphere}, orbits on which remain on a circle with constant radius $r_-$ and are unstable under radial perturbations.

\subsubsection{The normally hyperbolic invariant manifold}
Denote by $r_c=3M$ in the $\mu^2=0$ case and $r_c=r_-$ in the $\mu^2=1$ case.
 If we fix the constant $p_\tau^2=E^2=2V(r_c)$, which is a function of $L^2$ and $M$, then we can treat the Hamiltonian $H$ in \eqref{EqHamS} as a system of three degrees of freedom depending on the variables $r,p_r,\theta,p_\theta,\varphi,p_\varphi$. The four dimensional submanifold $$N:=\left\{r=r_c,\ \frac{dr}{dt}=0\right\}$$ is a normally hyperbolic invariant manifold (NHIM) in the sense of Definition \ref{DefNHIM}. One remarkable property of the NHIM is its persistence under small perturbations, which is summarized in Theorem \ref{ThmNHIM} in Appendix \ref{SNHIM}. The normal Lyapunov exponents are $\pm \sqrt{-V''(r_c)}$ obtained as the eigenvalues of the linearized radial Hamiltonian equation $$\left(\begin{array}{c} \dt r\\
\dt\dot r\end{array}\right)'=\left(\begin{array}{cc}0&1\\
-V''(r_c)&0\end{array}\right)\left(\begin{array}{c} \dt r\\
\dt\dot r\end{array}\right)$$
 at $r=r_c$. The coefficient matrix has two eigenvalues of opposite signs due to $V''(r_c)<0$. The Hamiltonian system restricted to $N$ is a Hamiltonian system of two degrees of freedom of the form
 \begin{equation}\label{HamN}H_N(\theta,p_\theta,\varphi,p_\varphi)=r_c^{-2}\left(p_\theta^2+\frac{1}{\sin^2\theta} L_z^2\right)-\al^{-1}(r_c) E^2.
\end{equation}
\subsection{Dynamics on the photon shell}\label{SSShellS}
The system \eqref{HamN} has two degrees of freedom and is integrable in the sense of Theorem \ref{ThmLA}. Note that this is the case for both $\mu^2=0$ and $\mu^2=1$. 

The nontrivial part of the Hamiltonian \eqref{HamN} is the square of the total angular momentum, denoted by $$L^2=p_\theta^2+\frac{1}{\sin^2\theta} L_z^2.$$
The dynamics of this Hamiltonian is as follows. First the Hamiltonian has no dependence on $\varphi$ so $L_z$ is a constant of motion. For fixed $L_z$, the resulting system describes only the latitudinal motion, which can be considered as a particle moving in the potential well of $\frac{1}{\sin^2\theta} L_z^2$, which is a  infinitely deep single well. Let  $\theta^-<\theta^+$ be two zeros of the $L^2- \frac{1}{\sin^2\theta} L_z^2$, then along an orbit, the latitudinal angle moves in the spherical strip $\theta^-\leq \theta\leq \theta^+$.

  We next introduce action-angle coordinates following Section 50 of \cite{A89}. We introduce the action variable
\begin{equation}\label{EqJtheta}J_\theta=\frac1{2\pi}\oint p_\theta\,d\theta=\frac{1}{\pi}\int_{\theta^-}^{\theta^+}\sqrt{L^2-\frac{1}{\sin^2\theta} L_z^2}\,d\theta.\end{equation}
 The other action variable is taken to be $L_z$ and their dual angular variables are denoted by $\al_\theta,\al_\varphi\in \T^1$ respectively, introduced to preserve the symplectic form
 $$dp_\theta\wedge d\theta+d L_z\wedge d\varphi=dJ_\theta\wedge d\al_\theta+dL_z\wedge d\al_\varphi.$$
In the new coordinates the Hamiltonian $H:=L^2(J_\theta,L_z):\ \R_+\times \T^1\times \R\times \T^1\to \R$ is independent of the angular variables $\al_\theta,\al_\varphi$. Then the Hamiltonian equation becomes $\begin{cases}\dot\al_\theta=\frac{\partial H}{\partial J_\theta}\\
\dot\al_\varphi=\frac{\partial H}{\partial L_z}
 \end{cases}$. Hence the Hamiltonian flow is a translation on $\T^2$ via $(\al_r,\al_\theta)\mapsto (\al_\theta,\al_\varphi)+ t\omega,\mathrm{mod}\ \Z^2,\ \ t\in \R$, where $\omega=(\omega_\theta,\omega_\varphi)=(\frac{\partial H}{\partial J_\theta},\frac{\partial H}{\partial L_z})$ is called the frequency.

When we restrict to the level set $H_N=-\frac12\mu^2$ and perform the energetic reduction procedure as we did in Section  \ref{SSERed} by taking $-J_\theta$ as the new Hamiltonian, we see that we get a the resulting time-1 map $\phi_0:\ \T\times I\to \T\times I$ via $(\al_\varphi, L_z)\mapsto (\al_\varphi+\nu, L_z)$,  where $\nu:=\omega_\varphi/\omega_\theta$ is exactly the rotation number in the integrable case, and $I=[\eps, |L|]$ for any small $\eps>0$. In order to apply Moser's theorem and the twist map theory we have to verify the twist condition (Definition \ref{DefTwist}(2)), which is exactly that the ratio $\omega_\varphi/\omega_\theta$ depends on its parameter $L_z$ strictly monotonically.

 The following result show that the map $\phi_0$ does not satisfy the twist condition, so neither Moser's theorem nor the twist map theory applies to a small perturbation of $\phi_0$ {\it a priori}.
\begin{Prop}\label{PropConvexityS}
The Hamiltonian $H_N$ written in action-angle coordinates satisfies
\begin{enumerate}
\item $\frac{\partial H_N}{\partial J_\theta}=\frac{\partial H_N}{\partial L_z}$;
\item for fixed $L_z>0$, as a function of $J_\theta$, the Hamiltonian is strictly convex in the region $J_\theta>0$.
\end{enumerate}
\end{Prop}

\begin{proof}
We first prove item (1).
  Let $f$ be the map relating the two sets of constants of motion $(J_\theta, L_z)$ and $(L^2,L_z)$, i.e.  $f(J_\theta, L_z)=(L^2,L_z)$. From $Df\cdot Df^{-1}=I$, we obtain
$$\left[\begin{array}{ccc}
\frac{\partial L^2}{\partial J_\theta}&\frac{\partial L^2}{\partial L_z}\\
\frac{\partial L_z}{\partial J_\theta}&\frac{\partial L_z}{\partial L_z}
\end{array}\right]\left[\begin{array}{ccc}
\frac{\partial J_\theta}{\partial L^2}&\frac{\partial J_\theta}{\partial L_z}\\
0&1
\end{array}\right]=I.$$
This implies $Df=\left[\begin{array}{ccc}
\left(\frac{\partial J_\theta}{\partial L^2}\right)^{-1}&- \frac{\partial J_\theta}{\partial L_z}\cdot\left(\frac{\partial J_\theta}{\partial L^2}\right)^{-1}\\
0&1
\end{array}\right].$ Denote by $$\omega:=DH(J_\theta,L_z)=\left( \left(\frac{\partial J_\theta}{\partial L^2}\right)^{-1}, - \frac{\partial J_\theta}{\partial L_z}\cdot\left(\frac{\partial J_\theta}{\partial L^2}\right)^{-1}\right)$$ the first row of $Df$.

We next show that $- \frac{\partial J_\theta}{\partial L_z}=1$ so that $\omega=\frac{\partial L^2}{\partial J_\theta}(1,1)$. Denoting by $\theta^-<\theta^+$ the two roots of $p_\theta$, then we have from \eqref{EqJtheta}
\begin{equation*}\label{EqAACoordS}
\begin{aligned}
\frac{\partial J_\theta}{\partial L^2}&=\frac{1}{2\pi}\int_{\theta^-}^{\theta^+} \left(L^2-\frac{L_z^2}{\sin^2 \theta}\right)^{-1/2} \,d\theta,\\
\frac{\partial J_\theta}{\partial L_z}&=-\frac{1}{\pi}\int_{\theta^-}^{\theta^+} \frac{L_z}{\sin^2 \theta}\left(L^2-\frac{L_z^2}{\sin^2 \theta}\right)^{-1/2} \,d\theta.\\
\end{aligned}
\end{equation*}
Note that the integrand of the last expression is exactly $\frac{d\varphi}{d\theta}$ (see \eqref{EqHamS}). Each Schwarzschild geodesic on the photon sphere lies on a fixed plane passing through the origin hence is a closed orbit, thus we get
$$\int_{0}^{2\pi}d\varphi=2\int_{\theta^-}^{\theta^+}\frac{L_zd\theta}{\sin^2\theta\sqrt{L^2-\frac{L_z^2}{\sin^2 \theta}}}.$$
This gives $\frac{\partial J_\theta}{\partial L_z}=-1$, hence  proves item (1) by applying the  implicit function theorem.

We next work on item (2). We develop Mather's variational method \cite{Ma1} to our new setting.
We consider the $(\theta,p_\theta)$-plane with $\theta\in [0,\pi],\ p_\theta\in \R$. The point $(\pi/2,0)$ is an elliptic fixed point. If we remove this point from the plane we get a topological cylinder on which we introduce the polar coordinates $F:(\rho,\psi)\in \R_+\times \T\mapsto (p_\theta,\theta,)$, i.e. $p_\theta=\sqrt\rho \sin\psi,\ \theta-\pi/2=\sqrt\rho \cos\psi$. On $\R_+\times \T$, we introduce the differential form $\eta_c=f(\psi)d\psi$, where  $f:\ \T\to \R_+$, with cohomology class $c=\int_{\T} f(\psi)\,d\psi$. It becomes a differential form on $[0,\pi]\times \R\setminus\{(\pi/2,0)\}$ after being pulled-back by $F^*$, still denoted by $\eta_c$.

We next introduce the space $\mathfrak M$ of probability measures defined on the phase space that are invariant under the Hamiltonian flow. 
For each $\mu\in \mathfrak M$, we introduce
$$A_c(\mu)=\int \mathbf L-\eta_c\,d\mu,$$
where $\mathbf L=\frac12\dot \theta^2+\frac12\dot\varphi^2\sin^2\theta$ is the formal Legendre transform of the Hamiltonian $H:=\frac{1}{2}L^2$.
For given $c=(c_\theta,c_\varphi)$ with $c_\theta>0$, we choose the differential 1-form $\eta_c$ to be $$p_\theta\,d\theta+L_zd\varphi:=\pm \sqrt{L^2-\frac{L_z^2}{\sin^2\theta}}\,d\theta+L_zd\varphi$$ such that $\frac1\pi\int \sqrt{L^2-\frac{L_z^2}{\sin^2\theta}}\,d\theta=c_\theta$ and  $c_\varphi=L_z$. Note that this fixes the constant $L^2$. Then we get
\begin{equation*}
\begin{aligned}
A_c(\mu)&=\int \mathbf L-\eta_c\,d\mu. \\
&=\int \mathbf L-p_\theta\cdot \dot \theta\,-L_z\dot \varphi d\mu\\
&=\int \frac12(\dot \theta-p_\theta)^2+\frac12(\dot\varphi \sin \theta-\frac{L_z}{\sin\theta})^2-\frac{L_z^2}{2\sin^2\theta}-\frac12p_\theta^2\,d\mu\\
&=-\frac12 L^2+\int \frac12(\dot \theta-p_\theta)^2+\frac12(\dot\varphi \sin \theta-\frac{L_z}{\sin\theta})^2\,d\mu.
\end{aligned}
\end{equation*}
If we minimize $A_c(\mu)$ among all $\mu\in \mathfrak M$, we get that the only choice is $-\al(c):=\inf_{\mu\in \mathfrak M}A_c(\mu)=-\frac12L^2$, and $\dot \theta= p_\theta,\ \dot\varphi =\frac{L_z}{\sin^2\theta}$. Note that the latter determines uniquely an invariant torus of the Hamiltonian flow for given $L_z$ and $L^2$, on which the minimizing measure $\mu$ supports.

Since we have the cohomology class $(c_\theta,L_z)$ is exactly the action variables $(J_\theta,L_z)$, we get that the function $\al$ as a function of the cohomology classes is exactly the total energy $\frac12 L^2$ as a function of the action variables $(J_\theta,L_z)$.
We next prove that $\al(c)$ is convex. Let $\mu_1$ and $\mu_2$ be such that $-\al(c_1)=A_{c_1}(\mu_1),\ -\al(c_2)=A_{c_2}(\mu_2)$ and $\mu_\lambda$ be such that $-\al(\lambda c_1+(1-\lambda)c_2)=A_{\lambda c_1+(1-\lambda)c_2}(\mu_\lambda)$. Then we have
\begin{equation*}
\begin{aligned}
-\al(\lambda c_1+(1-\lambda)c_2)&=A(\mu_\lambda)=\int\mathbf L-\eta_{\lambda c_1+(1-\lambda)c_2}\,d\mu_\lambda\\
&=\lambda \int\mathbf L-\eta_{c_1}\,d\mu_\lambda+(1-\lambda) \int\mathbf L-\eta_{c_2}\,d\mu_\lambda\\
&\geq \lambda \int\mathbf L-\eta_{c_1}\,d\mu_1+(1-\lambda) \int \mathbf L-\eta_{c_2}\,d\mu_2\\
&=-\lambda\al(c_1)-(1-\lambda)\al(c_2).
\end{aligned}
\end{equation*}
This shows that $\al$ is convex in $c$, hence the Hamiltonian $\frac12L^2$ is convex in $J_\theta, L_z$. It remains to prove the strict convexity in the direction transverse the energy levels. We introduce the $\beta$-function that is the Legendre dual of $\al$
$$\beta(h)=\sup_c\langle h,c\rangle-\al(c). $$
Following the argument of \cite{Ca}, we get that $\beta$ is differentiable in the radial direction. The Legendre transform gives $h=D\al(c)=(\frac{\partial H}{\partial  J_\theta}, \frac{\partial H}{\partial L_z})$ attaining the sup in the definition of $\beta$. By the first item of the Proposition, we get that $\al(c)=f(J_\theta+L_z)$ for some convex function $f$. Then for $L_z$ fixed, the differentiability of $\beta$ in the radial direction implies that $f$ is strictly convex. 

\end{proof}

The first item reflects the fact that all orbit on the photon sphere lies on a plane passing through the origin hence is periodic and they all have the same period on a fixed energy level. So neither KAM nor twist map theory can be applied at this moment. However, when a perturbation is added, it may introduce certain twistness, which require further more careful analysis depending on the form of the perturbation.


\subsection{KAM and QPO for Schwarzschild}
In the case of $\mu^2=1$, we have other bound orbits than those on the photon shell. In particular, orbits with $r=r_+$ are stable circular. In a neighborhood of this stable circular orbit, there are a lot of bound orbits.
 \subsubsection{KAM and twist maps theory for the QPO}
 We introduce the following bounded part of the phase space where Liouville-Arnold theorem applies. Let $C$ be a large constant and define
 \begin{equation}\label{EqBS}\mathcal{B}(C):=\left\{ V(r_+)< \frac{E^2}{2}<\min\{V(r_-),V(\infty)\},\quad |L_{isco}|<|L|<C\right\}\cap H^{-1}(-\frac12).\end{equation}


In this set $\mathcal{B}(C)$, we note that both the radial $r$-component and the latitudinal $\theta$-component are lying in potential wells around the local minimums, so nearby orbit will have small oscillations in the $r$- and $\theta$-components (c.f. Figure \ref{FigV}).

Suppose for simplicity that we consider a stationary and axisymmetric perturbation of the Schwarzchild metric so  that the perturbation has no dependence on $\varphi$ and $\tau$ then $E$ and $L_z$ are constants of motion. Fixing a value for $E$ and $L_z$ satisfying the definition of $\mathcal{B}(C)$, then the resulting Hamiltonian system has  two degrees of freedom with coordinates $(r,p_r,\theta,p_\theta)$.

We next introduce the action-angle coordinates and work out the frequencies of the oscillations. 
The action variables $J_r$ and $J_\theta$ are defined as follows (c.f. \eqref{EqHamS})
\begin{equation}\label{EqAAS}
\begin{aligned}
J_r&=\frac{1}{2\pi}\oint p_r dr=\frac1\pi\int_{r^-}^{r^+}\frac{1}{\al}\sqrt{E^2-\al r^{-2} L^2-\al} \,dr,\\
J_\theta&=\frac{1}{2\pi}\oint p_\theta d\theta=\frac{1}{\pi}\int_{\theta^-}^{\theta^+}\sqrt{L^2-\frac{L_z^2}{\sin^2\theta}}\,d\theta ,
\end{aligned}
\end{equation}
where $r^-<r^+$ are zeros of the integrand of $J_r$, and the definition of $J_\theta$ is the same as \eqref{EqJtheta}. The dual angular coordinates $\al_r,\al_\theta\in \T$ are introduced to preserve the symplectic structure, i.e.
$$dp_r\wedge dr+dp_\theta\wedge d\theta=dJ_r\wedge d\al_r+dJ_\theta\wedge d\al_\theta.$$ The motion of the angular coordinates $(\al_r,\al_\theta)$ on $\T^2$ is given by a linear motion with frequency $\omega=(\omega_r,\omega_\theta)=(\frac{\partial H}{\partial J_r},\frac{\partial H}{\partial J_\theta})$ where $H$ is the Schwarzschild Hamiltonian written in action-angle coordinates and with fixed parameters $E$ and $L_z$. We call $\omega_r:=\frac{\partial H}{\partial J_r}$ the radial frequency and $\omega_\theta:=\frac{\partial H}{\partial J_\theta} $ the latitudinal frequency. 

The ratio $\omega_r/\omega_\theta$ can be parametrized by one parameter that we choose to be $L^2$. In order to apply KAM theorem and the twist map theory, we need to find an interval where the ratio has no critical point.  We see from Figure \ref{FigQPOS} (plotted with $E=0.97$ and $L_z=3$) that this is indeed the case.

\begin{Thm}\label{ThmQPOS} Let $\eps h_{\mu\nu} dx^\mu dx^\nu$ be a stationary and axisymmetric perturbation to the Schwarzschild metric. Then there exist an open set of parameters $(E,L_z)$,  a neighborhood $\mathcal D$ of the point $(r,p_r,\theta,p_\theta)=(r_+,0,\pi/2, 0)$ and $\eps_0$ such that KAM and the twist map theory are applicable to the perturbed system in $\mathcal D$ and for all $0<|\eps|<\eps_0$.
\end{Thm}

\begin{figure}
\centering
\includegraphics[width=0.4\textwidth]{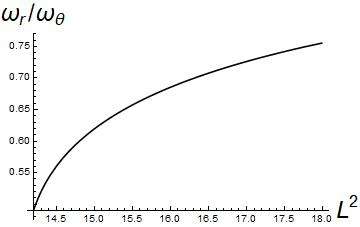}
\caption{The ratio $\omega_r/\omega_\theta$ for Schwarzschild with parameters $E=0.97$ and $L_z=3$}
\label{FigQPOS}
\end{figure}

\begin{proof}
Let $f$ be the map relating the two sets of constants of motion $(J_r,J_\theta)$ and $(H,L^2)$: $f(J_r,J_\theta)=(H,L^2)$, where $H$ is the Schwarzchild Hamiltonian, then we have $Df\cdot Df^{-1}=I$, which is more explicitly
$$\left[\begin{array}{ccc}
\frac{\partial H}{\partial J_r}&\frac{\partial H}{\partial J_\theta}\\
\frac{\partial L^2}{\partial J_r}&\frac{\partial L^2}{\partial J_\theta}
\end{array}\right]\left[\begin{array}{ccc}
\frac{\partial J_r}{\partial H}&\frac{\partial J_r}{\partial L^2}\\
0&\frac{\partial J_\theta}{\partial L^2}
\end{array}\right]=I.$$
Here we use the fact that $J_\theta$ is independent of $H$ as we see from the expression of $J_\theta$.
This implies $Df=\left[\begin{array}{ccc}
\left(\frac{\partial J_r}{\partial H}\right)^{-1}&- \frac{\partial J_r}{\partial L^2}\cdot\left(\frac{\partial J_r}{\partial H}\frac{\partial J_\theta}{\partial L^2}\right)^{-1}\\
0&\left(\frac{\partial J_\theta}{\partial L^2}\right)^{-1}
\end{array}\right].$  Denote by $$\omega=(\omega_r,\omega_\theta)=\left(\left(\frac{\partial J_r}{\partial H}\right)^{-1},- \frac{\partial J_r}{\partial L^2}\cdot\left(\frac{\partial J_r}{\partial H}\frac{\partial J_\theta}{\partial L^2}\right)^{-1}\right)$$ the first row of $Df$. 
Denoting $P(r)=E^2-\al r^{-2} L^2-\al$, we get from \eqref{EqAAS}
\begin{equation*}\label{EqDJ}
\begin{aligned}
\frac{\partial J_r}{\partial L^2}&=-\frac{1}{2\pi}\int_{r^-}^{r^+}\frac{1}{r^2\sqrt{P(r)}}\,dr,\\
\frac{\partial J_\theta}{\partial L^2}&=\frac{1}{2\pi}\int_{\theta^-}^{\theta^+}\frac{1}{\sqrt{L^2-\frac{L_z^2}{\sin^2\theta}}}\,d\theta.
\end{aligned}
\end{equation*}

So we get $$\frac{\omega_r}{\omega_\theta}=-\frac{\partial J_\theta}{\partial L^2}\cdot\left(\frac{\partial J_r}{\partial L^2}\right)^{-1}=\int_{\theta^-}^{\theta^+}\frac{1}{\sqrt{L^2-\frac{L_z^2}{\sin^2\theta}}}\,d\theta\Big/\int_{r^-}^{r^+}\frac{1}{r^2\sqrt{P(r)}}\,dr,$$ which is a function of $L^2$ from the above formulas. Figure \ref{FigQPOS} is plotted using this formula.
Since we have $\frac{\partial L^2}{\partial J_r}=0,\frac{\partial L^2}{\partial J_\theta}=(\frac{\partial J_\theta}{\partial L^2})^{-1}$, we can then eventual compute
$\frac{\partial}{\partial J_\theta}\frac{\omega_r}{\omega_\theta}.$ We see from Figure  \ref{FigQPOS}  that the statement holds for a particularly chosen paramter $(E,L_z)$. Since the action-angle coordinate change is smooth, so we get the statement for the parameter $(E,L_z)$ in an open set. The proof of the statement is now complete.

\end{proof}
\subsubsection{KAM nondegeneracy and further derivatives}\label{SSSDerivatives}
The preceding Theorem \ref{ThmQPOS} uses Moser's version of KAM theorem. If we would like to verify the KAM nondegeneracy conditions in Section \ref{SSKAM}, we have to calculate $\nabla \omega$ where $\nabla $ is the derivative with respect to $J_r,J_\theta$.  In this section, for the sake of completeness, we  show how to take further derivative of $\omega$ and clarify some subtleties. Since we do not use it in the proof of the main theorem, readers are suggested to skip this section when first reading.

We have
\begin{equation}
\begin{aligned}\nabla\omega&=\frac{\partial(E^2,L^2)}{\partial(J_r,J_\theta)}(\partial_{E^2},\partial_{L^2})\omega\\
&=Df \cdot\left[\begin{array}{ccc}
\partial_{E^2}\left(\frac{\partial J_r}{\partial E^2}\right)^{-1}&\partial_{L^2}\left(\frac{\partial J_r}{\partial E^2}\right)^{-1}\\
-\partial_{E^2}\left(\frac{\partial J_r}{\partial L^2}\cdot\left(\frac{\partial J_r}{\partial E^2}\frac{\partial J_\theta}{\partial L^2}\right)^{-1}\right)&- \partial_{L^2}\left(\frac{\partial J_r}{\partial L^2}\cdot\left(\frac{\partial J_r}{\partial E^2}\frac{\partial J_\theta}{\partial L^2}\right)^{-1}\right)
\end{array}\right]\\
&=Df \cdot  \left(\frac{\partial J_r}{\partial E^2}\right)^{-3}\left(\frac{\partial J_\theta}{\partial L^2}\right)^{-1}\left[\begin{array}{ll}
 \frac{\partial^2 J_r}{(\partial E^2)^2}& \frac{\partial^2 J_r}{\partial L^2\partial E^2}\\
 \frac{\partial^2 J_r}{\partial E^2\partial L^2}&\frac{\partial^2 J_r}{(\partial L^2)^2}-\frac{\partial^2 J_\theta}{(\partial L^2)^2}\frac{\partial J_r}{\partial L^2}\cdot \left(\frac{\partial J_\theta}{\partial L^2}\right)^{-1}
\end{array}\right].
\end{aligned}\end{equation}
The expressions in the last matrix are all computable from the definitions of $J_r,J_\theta$.

To obtain the second order derivatives, we cannot differentiate the first order derivatives in \eqref{EqDJ} for the technical reason that further derivative will introduce divergent improper integrals.  This problem can be resolved as follows. We consider only the case $J_r$ and the case $J_\theta$ is similar. Note that the singularity originates from the fact that $P(r^\pm)=0$ will appear in the denominator, which we should avoid. The action $J_r$ is the area enclosed by the curve Graph$p_r$. Let $r_*$ be the minimum of $\al r^{-2} L^2+2\al$. Then we split the integral $J_r=\int_{r^-}^{r^+}\frac{1}{\pi\al}\sqrt{P(r)}\,dr$ into two pieces $J_r^1=\int_{r^-}^{r_*}\frac{1}{\pi\al}\sqrt{P(r)}\,dr$ and $J_r^2=\int_{r_*}^{r^+}\frac{1}{\pi\al}\sqrt{P(r)}\,dr$ and consider only $J_r^1$ since the other one is completely similar. We next pick $\tilde r\in (r^-, r_*)$ sufficiently close to $r^-$ and denote by $\tilde p$ the positive value of $p_r(\tilde r)$. We denote $ J_r^{1,R}=\int_{\tilde r}^{r_*}\frac{1}{\pi\al}\sqrt{P(r)}\,dr$ the area enclosed by $r=\tilde r, r=r_*$ and Graph$p_r$. For the area to the left of the line $r=\tilde r$, we realize the Graph$p_r$ to the left of $r=\tilde r$ as a graph over the $p_r$-axis. So we need to solve the function $\al^2 p_r^2=P(r(p_r))$ for $r$ as a function of $p_r$. This function is solvable if $\tilde r$ is sufficiently close to $r_-$ since $P'(r^-)\neq 0$ hence the implicit function theorem applies. With this function $r(p_r)$, we introduce the integral $J^{1,L}_r=\frac{1}{\pi}\int_{-\tilde p}^{\tilde p}r(p_r) dp_r$. The overlap of $J^{1,L}_r$ and $J^{1,R}_r$ is a square of area $2\tilde p|r_*-\tilde r|. $ Therefore we obtain
$$J^{1}_r=J^{1,R}_r+J^{1,L}_r-2\tilde p|r_*-\tilde r|=\int_{\tilde r}^{r_*}\frac{1}{\pi\al}\sqrt{P(r)}\,dr+\frac{1}{\pi}\int_{-\tilde p}^{\tilde p}r(p_r) dp_r-2\tilde p|r_*-\tilde r|.$$
We remark that $\tilde p$ and $r(p_r)$ are both dependent on the variables $E^2$ and $L^2$ while $\tilde r$ is not.  The derivatives of this expression does not involve improper integrals hence we can take its second derivative to yield an expression for $D\omega$.

\subsubsection{Three fundamental frequencies}\label{SSS3Freq}
In Theorem \ref{ThmQPOS} we consider stationary and axisymmetric perturbations so we have only two fundamental frequencies $\omega_r$ and $\omega_\theta$ to consider. In general, when a perturbation depends on $\varphi$, i.e. nonaxisymmetric, but not on $\tau$, we have to treat the Hamiltonian system with as one three degrees of freedom, in which case, we cannot  apply Moser's theorem or twist map theory, so we have to verify the KAM nondegeneracy assumptions in Section \ref{SSKAM} if we wish to obtain KAM type dynamics. Again readers are suggested to skip this section when first reading, since here we only show how to get the formulas for $\nabla \omega$, where $\nabla$ means the derivative with respect to the action variables $J_r,J_\theta,J_\varphi$. We introduce the action variables $J_r,J_\theta$ as before and the extra $J_\varphi$ is chosen to be $L_z$.

Then there is a map $F$ relating the two set of conserved quantities: $F(J_r,J_\theta,J_\varphi)=(H, L^2,L_z)$. We have $DF\cdot DF^{-1}=I$. More explicitly we have
$$\left[\begin{array}{ccc}
\frac{\partial H}{\partial J_r}&\frac{\partial H}{\partial J_\theta}&\frac{\partial H}{\partial J_\varphi}\\
\frac{\partial L^2}{\partial J_r}&\frac{\partial L^2}{\partial J_\theta}&\frac{\partial L^2}{\partial J_\varphi}\\
0&0&1
\end{array}\right]\left[\begin{array}{ccc}
\frac{\partial J_r}{\partial H}&\frac{\partial J_r}{\partial L^2}&0\\
0&\frac{\partial J_\theta}{\partial L^2}&\frac{\partial J_\theta}{\partial L_z}\\
0&0&1
\end{array}\right]=I.$$
This implies $DF=\left[\begin{array}{ccc}
a^{-1}&-b/(ac)&bd/(ac)\\
0&c^{-1}&-d/c\\
0&0&1
\end{array}\right]$, where $a=\frac{\partial J_r}{\partial H},\ b=\frac{\partial J_r}{\partial L^2},\ c=\frac{\partial J_\theta}{\partial L^2},\ d=\frac{\partial J_\theta}{\partial L_z}$ can be evaluated directly from equation \ref{EqAAS}.

Denoting by $\omega=\omega(J_r,J_\theta,J_\varphi)$ the first row of $DF$. We thus have
$$\nabla\omega=\frac{\partial(H,L^2,L_z)}{\partial(J_r,J_\theta,J_\varphi)}(\partial_H,\partial_{L^2},\partial_{L_z})\omega=DF \cdot\left[\begin{array}{ccc}
\partial_H(a^{-1})&\partial_{L^2}(a^{-1})&0\\
\partial_H(-b/(ac))&\partial_{L^2}(-b/(ac))&\partial_{L_z}(-b/(ac))\\
\partial_H(bd/(ac))&\partial_{L^2}(bd/(ac))&\partial_{L_z}(bd/(ac))\\
\end{array}\right]$$
Since $DF$ is nondegenerate, the nondegeneracy condition is then reduced to that of $(\partial_E,\partial_{L^2},\partial_{L_z})\omega$.

\section{The Kerr spacetime}\label{SSK}
In this section, we study the dynamics of perturbed Kerr spacetime.
The Kerr spacetime in the standard Boyer-Lindquist coordinates has the form
$$ds^2=-\left( 1-\frac{2Mr}{\Sigma}\right)d\tau^2-2\left( \frac{2Mr}{\Sigma}\right) a\sin^2\theta d\tau d\varphi+\Sigma \left( \frac{dr^2}{\Delta}+d\theta^2\right) +\frac{\mathcal A}{\Sigma}\sin^2\theta d\varphi^2,$$
where
\begin{equation}
\begin{aligned}
\Sigma&=r^2+a^2\cos^2\theta,\\
\Delta&=r^2+a^2-2Mr,\\
\mathcal A&=(r^2+a^2)^2-\Delta a^2\sin^2\theta.
\end{aligned}
\end{equation}
Here $M$ is the mass and $a$ is the angular momentum of the blackhole with $0\leq|a|\leq M$. When $a=0$, the Kerr metric reduces to the Schwarzschild metric.
The event horizon is where the metric coefficient $g_{rr}=\frac{\Sigma}{\Delta}$ becomes singular. Let $r_h=M+ \sqrt{M^2-a^2}$ be the larger root of $\Delta$, hence the outer event horizon is the sphere $r=r_h$.


 The geodesic equations are of the following form (assuming $ds^2=-\mu^2$).
\begin{equation}\label{EqofMotion}
\begin{cases}
\Sigma^2(\frac{dr}{dt})^2&= R,\\
\Sigma^2(\frac{d\theta}{dt})^2&=\Theta,\\
\Sigma\frac{d\tau}{dt}&=-a(aE\sin^2\theta-L_z)+\frac{r^2+a^2}{\Delta}P(r),\\
\Sigma\frac{d\varphi}{dt}&=\frac{a}{\Delta}P(r)-\left(aE-\frac{L_z}{\sin^2\theta}\right),
\end{cases}
\end{equation}
where we have
\begin{equation}\label{EqRTheta}
\begin{aligned}
&R(r)= P(r)^2-\Delta(\mu^2 r^2+(L_z-aE)^2+Q),\\
&\Theta(\theta)=Q-\left[(\mu^2-E^2)a^2+\dfrac{L_z^2}{\sin^2\theta}\right]\cos^2\theta,\\
&P(r)=E(r^2+a^2)-aL_z,
\end{aligned}
\end{equation}
and the constant $Q=p_\theta^2+\cos^2\theta\left(a^2(\mu^2-E^2)+\left(\frac{L_z}{\sin\theta}\right)^2\right)$ is a constant of motion called \emph{Carter constant}.

We next introduce the Hamiltonian formalism. We treat the Kerr metric as twice a Lagrangian and get the corresponding Hamiltonian via the formal Legendre transform $H(p_{x^\mu},x^\mu)=p_{x^\mu}\frac{d x^\mu}{dt}-L(x^\mu,\frac{dx^\mu}{dt})$ where $x^\mu=(\tau, r,\theta,\varphi)$ and $p_{x^\mu}=(p_\tau,p_r, p_\theta,p_\varphi)$ with
\begin{equation}\label{EqLegendreK}
\begin{cases}
p_\tau&=E=-\left( 1-\frac{2Mr}{\Sigma}\right)\dot\tau-\left( \frac{2Mr}{\Sigma}\right) a\sin^2\theta  \dot\varphi,\\
p_r&=\Sigma \frac{\dot r}{\Delta},\\
 p_\theta&=\Sigma\dot \theta,\\
 p_\varphi&=L_z=-\left( \frac{2Mr}{\Sigma}\right) a\sin^2\theta \dot\tau +\frac{\mathcal A}{\Sigma}\sin^2\theta \dot \varphi.
 \end{cases}
 \end{equation}
 The Hamiltonian $H$  can be written as
\begin{equation}\label{EqHam}
\begin{aligned}
& H=\frac12\Sigma^{-1}(H_r+H_\theta),\\ 
& H_r=\Delta p_r^2-\dfrac{1}{\Delta}((r^2+a^2)E-aL_z)^2,\\
& H_\theta=p_\theta^2+\dfrac{1}{\sin^2\theta}(a\sin^2\theta E-L_z)^2.\\
\end{aligned}
\end{equation}
The system has four independent constants of motion $H,E,L_z,Q$. 

\subsection{The vertical dynamics}\label{SSSVertical}
Introducing a time reparametrization $d\lambda=\Sigma(r,\theta)^{-1}dt$ called Mino time, we determine the vertical dynamics  by the equation \begin{equation}\label{EqVerticalK}(\frac{d\theta}{d\lambda})^2=\Theta=Q-U(\theta),\ \mathrm{where}\ U(\theta)=\left[(\mu^2-E^2)a^2+\dfrac{L_z^2}{\sin^2\theta}\right]\cos^2\theta.\end{equation} So the dynamics can be visualized as a particle moving in the potential well of $U(\theta)$.

The case of $L_z=0$ differs drastically from $L_z\neq 0$ case. Indeed, if $L_z=0$ and $Q>(\mu^2-E)a^2$, the orbit goes through the south and north poles,  while $L_z\neq0$, as will will see in the following, the orbit is bounded away from the the south and north poles. 
\begin{figure}
\centering
\includegraphics[width=0.3\textwidth]{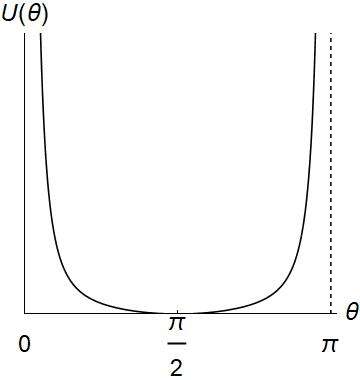}
\caption{Graph of $U$}
\label{FigU}
\end{figure}

We first consider the case $\mu^2=1$ and suppose  $(\mu^2-E^2)>0$, then we have $U(\theta)\geq 0$ is an infinitely deep single well potential and $Q\geq 0$ (see Figure \ref{FigU}). The minimum $\min U=0$ is attained at $\theta=\pi/2$. We also have $U(\pi/2+\theta)=U(\pi/2-\theta)$ for $\theta\in [0,\pi/2)$ and $U(\theta)\to\infty$ for $\theta\to 0,\pi.$ Let $\theta^-<\theta^+$ be two root of the equation $\Theta(\theta)=0$ with $Q>0$. Then the orbit on the photon sphere with conserved quantities $Q,E,L_z$ oscillates within the spherical strip $\theta\in [\theta^-,\theta^+]$. The explicit solution of $\theta(\lambda)$ is obtained as follows. Denoting by $$\ell=L_z/E,\ q=Q/E^2,\ v=\cos\theta,\ \al^2=a^2(E^{-2}-1)>0$$ and $$\Theta_v=q-(\ell^2+q+\al^2)v^2+\al^2 v^4,$$
and by $v^2_\pm$ the two roots of the equation $\Theta_v=0$ with $v_-^2<v_+^2$.  Then we get
$$\lambda-\nu=\int_0^\theta \frac{d\theta}{\sqrt\Theta}=E^{-1}\int_0^v \frac{dv}{\sqrt\Theta_v}=\frac{1}{E\al}\int_0^v\frac{dv}{\sqrt{(v^2-v_+^2)(v^2-v_-^2)}}=\frac{1}{E\al v_+} F(\psi,k),$$
where $\nu$ is the integral constant, $k=\frac{v_-}{v_+}$ and $\sin \psi=\frac{v}{v_-}$, and $F$ is the incomplete elliptic integral of the first kind. This gives us the function of $\lambda$ in terms of $v$ hence of $\theta$, then inverting it we obtain the function $\theta(\lambda)$.

We next consider the  case  of $\mu^2=0$.
\begin{Lm}
\begin{enumerate}
\item When $a^2E^2\leq L_z^2$, the potential $U$ is a single well with a global minimum value at $\theta=\pi/2$;
\item when $a^2E^2>L_z^2$, the potential $U$ is a double well with a local max value 0 at $\theta=\pi/2$.
\end{enumerate}
\end{Lm}
\begin{proof}
When $a^2E^2<L_z^2$ we have $U\geq 0$ with a global mininum value 0 at $\theta=\pi/2$. So the situation is similar to the above $\mu^2=1$ and  $E^2<1$ case. It is necessary that $Q\geq 0$ and when $Q>0$ all orbits crosses the equator plane. On the other hand, we have that
$$U''(\pi/2)=4(L_z^2-E^2a^2).$$
This shows that when $L_z^2<E^2a^2$, the point $\theta=\pi/2$ is a local max of $U$. Since we have $U(\pi/2)=0$. This implies that the potential $U$ is negative in some region between $0$ and $\pi/2$, and between $\pi/2$ and $\pi$. Hence $U$ is a double well potential. In this case, it is possible to have $Q<0$, in  which case each orbit is confined within one well bounded away from the equator plane. When $Q>0$, the particle can move from one well to the other crossing the equator plane. The $Q=0$ case is critical, in which case the point $\theta=\pi/2$ corresponds to a hyperbolic fixed point and the orbit moving in one well approaches the fixed point in both future and pass hence is on a homoclinic orbit.
\end{proof}

 In this critical case $Q=0$ and the double well case, using the general principle of \cite{X}, in a generically periodically perturbed Kerr metric, separatix splitting will be created and chaotic motion will occur, which implies that there exists orbit tunneling from one well to the other.   In each of the above cases, solutions $\theta(\lambda)$ of the equation of motion can be found explicitly (c.f. Section 63 of \cite{C}).

\subsection{The photon shell and the radial dynamics}

 Similar to the photon sphere in Schwarzschild case, the Kerr spacetime also admits unstable spherical orbits. However, the dynamics of these orbits in the Kerr spacetime differs drastrically from the Schwarzschild case. Such an orbit in the Schwarzschild case lies on a fixed plane through the origin thus is periodic and all these orbits form a sphere. However, in the Kerr case, such an orbit is no longer  confined to a plane but may oscillate in a spherical strip and may be quasiperiodic. Moreover, depending on the parameters $E,L_z, Q$, the radius varies, so the union of all such orbits forms a ring called \emph{photon shell}.
 The photon shell for Kerr was first studied by \cite{Wi} in the extreme case $a=M$ for timelike geodesics. The lightlike case was studied by \cite{T}. The equatorial case of $Q=0$ was studied by \cite{GLP}. 




\subsubsection{Determining the photon shell}
An orbit in the photon shell have constant radial component. So from the radial equation of \eqref{EqofMotion}, the following equations should be satisfied.
\begin{equation}\label{EqCrit}R(r)=0,\quad \frac{d}{dr}R(r)=0. \end{equation}
Let $r_c$ be a solution. Moreover, orbits in the photon shell are unstable, which implies the second order derivative $\frac{d^2}{dr^2}R(r_c)>0$.

Introducing $\ell=\frac{L_z}{E}$ and $q=\frac{Q}{E^2}$, we write explicitly equations \eqref{EqCrit} as
\begin{equation*}
\begin{aligned}
& r^4+(a^2-\ell^2-q)r^2+2M(q+(\ell-a)^2)r-a^2q=0,\\
& 4r^3+2(a^2-\ell^2-q)r+2M(q+(\ell-a)^2)=0
\end{aligned}
\end{equation*}
in the $\mu^2=0$ case and
\begin{equation*}
\begin{aligned}
& 3r^4+r^2a^2-q(r^2-a^2)-\frac{r^2}{E^2}(3r^2-4Mr+a^2)=r^2\ell^2,\\
& r^4-a^2Mr+q(a^2-Mr)-\frac{r^3}{E^2} (r-M)=rM(\ell^2-2a\ell)
\end{aligned}
\end{equation*}
in the $\mu^2=1$ case.

Each pair of equations involves four variables $r,E,\ell,q$. We solve for $\ell$ and $q$
\begin{equation}\label{Eqellq0}
\begin{aligned}
\ell&=\frac{1}{a(r-M)}(M(r^2-a^2)-r\Delta),\\
q &=-\frac{r^3}{a^2(r-M)^2}(4M\Delta-r(r-M)^2)
\end{aligned}
\end{equation}

in the $\mu^2=0$ case and
\begin{equation}\label{Eqellq}
\begin{aligned}
\ell&=\frac{1}{a(r-M)}\left( M(r^2-a^2)- r\Delta \left( 1-\frac{1}{E^2}(1-\frac{M}{r})\right)^{1/2}\right),\\
q &=\frac{1}{a^2(r-M)}\left[\frac{r^3}{r-M}(4a^2M-r(r-3M)^2)+\frac{r^2}{E^2}\left[r(r-2M)^2-a^2M\right]\right.\\
&\left.-\frac{2r^3M}{r-M}\Delta\left(1- \left(1-\frac{1}{E^2}(1-\frac{M}{r})\right)^{1/2}\right)\right]
\end{aligned}
\end{equation}
in the $\mu^2=1$ case. The formulas can be found in Section 63 and 64 of \cite{C} and we have discarded an unphysical solution in each case. 

The functions $\ell$ and $q$ in  \eqref{Eqellq} are plotted in Figure \ref{fig:side:aChandra} with data $M=1, a=0.8, E=1$ and Figure \ref{fig:side:bChandra} with $M=1, a=0.8, E=0.92$ in the case of $\mu^2=1$.  For nearby choices of parameters the picture is qualitatively similar since the functions are continuous in $E^2$ and $a$. In the massless case, the RHS of \eqref{Eqellq0} has no dependence on $E$. The picture of $\ell$ and $q$ in \eqref{Eqellq0} is qualitatively similar to Figure \ref{fig:side:aChandra} for $a=0.8$. 

\begin{figure}
\begin{minipage}[t]{0.5\linewidth}
\centering
\includegraphics[width=2.in]{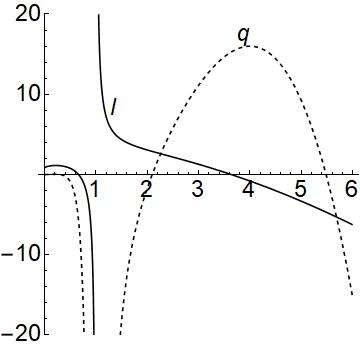}
\subcaption{$M=1,a=0.8,E=1$}
\label{fig:side:aChandra}
\end{minipage}%
\begin{minipage}[t]{0.5\linewidth}
\centering
\includegraphics[width=2.in]{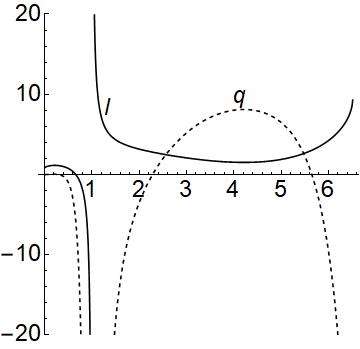}
\subcaption{$M=1,a=0.8,E=0.92$}
\label{fig:side:bChandra}
\end{minipage}
\caption{Graphs of $\ell$ and $q$ in the $\mu^2=1$ case}
\end{figure}


Given $r=r_c$, it determines the values of $\ell=\ell_c,q=q_c$ such that $r_c$ solves equation \eqref{EqCrit} with parameters $\ell_c$ and $q_c$. So for this choice of parameters $\ell_c, q_c$, the radius $r_c$ is a double root of $R(r)=0$, thus $R$ can be put in the form
\begin{equation}\label{EqRFactor}
R(r)/E^2=(r-r_c)^2\left(Ar^2 +2rB+C\right),\end{equation}
where we define $A=1-1/E^2<0, B=r_c\left( 1-\frac{1}{E^2}\left(1-\frac{M}{r_c}\right)\right), C=-\frac{a^2}{r_c^2} q_c$ in the $\mu^2=1$ case,  and $A=1, B=r_c,C=-a^2q_c/r_c^2$ in the $\mu^2=0$ case.

We next introduce the following admissible set of radii of bound spherical orbits.

\begin{Def}\label{DefAdmE}
\begin{enumerate}
\item In the case of $\mu^2=1$, let $M>0,\ |a|\in (0,M]$ and $E\in (1-Mr_h^{-1},1)$ be given parameters. We define the admissible set $\mathcal R:=\mathcal R(M,a,E)$ as the radius $r_c$ satisfying the following:
\begin{enumerate}
\item $r_h<r_c\leq \frac{M}{1-E^2}$ so that $r_c$ lies outside the event horizon and the square root in \eqref{Eqellq} makes sense;
\item $q(r_c)\geq 0;$
\item inequality $Ar_c^2+2Br_c+C>0$ holds, so that $r_c$ is a local max for $-R$.
\end{enumerate}
\item In the case of $\mu^2=0$, the admissible set $\mathcal R:=\mathcal R(M,a)$ radii of bound spherical orbits is defined as the set of  $r_c$ satisfying $q(r_c)\geq 0$ in \eqref{Eqellq0}.
\end{enumerate}
\end{Def}

\subsubsection{The Hamiltonian restricted to the photon shell}
When restricted to the photon shell \begin{equation}
\label{EqNHIMK}N:=\{r=r_c,p_r=0,\ r_c\in \mathcal R\}
\end{equation}
and fixing an admissible energy $E$, the Hamiltonian system becomes the Hamiltonian subsystem
\begin{equation}\label{EqHamPhotonK}
\begin{aligned}
& H_N=\Sigma(r_c,\theta)^{-1}\left(-\dfrac{((r_c^2+a^2)E-aL_z)^2}{\Delta(r_c)}+p_\theta^2+\dfrac{1}{\sin^2\theta}(a\sin^2\theta E-L_z)^2\right),\\
\end{aligned}
\end{equation}
 obtained from the Hamiltonian \eqref{EqHam} by setting $r=r_c, p_r=0$. We should further eliminate the $r_c$-dependence by substituting $r_c$ as a function of $L_z$ by inverting the first equation in \eqref{Eqellq0} or \eqref{Eqellq} on each of its monotone intervals.  Thus we get a Hamiltonian system of two degrees of freedom with coordinates $(\theta, p_\theta, \varphi, L_z)$ integrable in the Liouville-Arnold sense.

\subsubsection{The radial  dynamics}
We look at the radial motion in \eqref{EqofMotion}. After a time rescaling $d\lambda:=\frac{1}{\Sigma} d\tau$ called the Mino time, the radial equation of motion becomes $\left(\frac{dr}{d\lambda}\right)^2- R(r)=0$. Thus  the dynamics  can be visualized as a particle moving in the potential well of $-R(r)$ on the zero energy level.

\begin{figure}
\begin{minipage}[t]{0.5\linewidth}
\centering
\includegraphics[width=1.8in]{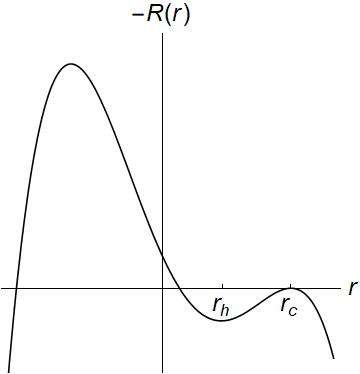}
\subcaption{$\mu^2=0$}
\label{FigR0}
\end{minipage}%
\begin{minipage}[t]{0.5\linewidth}
\centering
\includegraphics[width=1.8in]{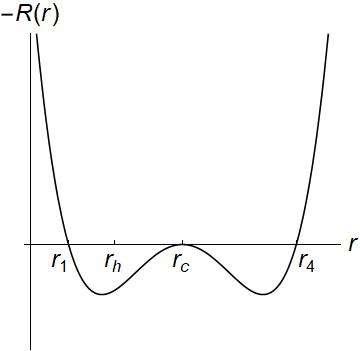}
\subcaption{$\mu^2=1$}
\label{FigR}
\end{minipage}
\caption{Graph of $-R$}
\end{figure}

\begin{Lm}In the case of $\mu^2=0$, the only bound orbits outside the event horizon are on the photon shell, i.e. we have $r=r_c$ along each such  orbit and $-R''(r_c)<0$.
\end{Lm}
\begin{proof}
 Denote by $\bar r<\bar{\bar r} $ the two roots of the factor $Ar^2+2Br+C $ in \eqref{EqRFactor}. Note that  $AC<0$, so we have $\bar r<0$. The lemma follows from the fact that $\bar{\bar r}<r_c$ (see Figure \ref{FigR0}). To prove it, it is enough  to show that $Ar_c^2+2Br_c+C>0$, which is equivalent to $3r_c^2>\frac{a^2 q}{r_c^2}$.   Using \eqref{Eqellq0}, the inequality is further reduced to $r_c(r_c-M)^2>M(r_c^2-2Mr_c+a^2)$ in the  case of $r_c> M$, which is  clearly true.  If $r_c=M$, the inequality is still true if $|a|<M$.  In  the extreme case $|a|=M$, the  event horizon is  $r=M$ so  the inequality is only violated when $r_c\leq M$  that means that the photon shell is not outside the event horizon. It can be checked that when $|a|=M$ and $r_c\leq M$, there is no bound  orbit lying  outside the event horizon.
\end{proof}

We next consider the $\mu^2=1$ case. Consider an admissible $r_c\in \mathcal R$, then we have \[R(r)=(E^2-1)(r-r_1)(r-r_c)^2(r-r_3),\ r_1<r_c<r_3.\]
See Figure \ref{FigR}.
Besides the bound spherical orbit, there is also a homoclinic orbit with $r(\lambda)$ approaching $r_c$ as $\lambda\to\pm\infty$. We do not treat the dynamics of the homoclinic orbit in this paper but refer readers to \cite{X} for more details.

\subsection{KAM and the photon shell dynamics}
We next consider the dynamics on the photon shell in Kerr spacetime. Restricted to the NHIM $N$ (c.f. \eqref{EqNHIMK}), the Hamiltonian $H_N$ is \eqref{EqHamPhotonK} with two degrees of freedom. Each orbit has constant radius hence lies on a sphere. The azimuthal angle $\varphi$ rotates periodically in $[0,2\pi)$ and the latitudinal angle $\theta$ oscillates within a spherical band symmetric around the equator according to \eqref{EqVerticalK}.



\subsubsection{KAM and twist map nondegeneracy}
We next introduce action-angle coordinates for the Hamiltonian $H_N$. In the equation $\Sigma^2\dot \theta^2=\Theta$ (c.f. \eqref{EqofMotion}), we substitute  $p_\theta=\Sigma\dot \theta$ (c.f. \eqref{EqLegendreK}) to yield $p_\theta^2=\Theta$. We consider both cases $\mu^2=0$ and $\mu^2=1$ together, and always assume $Q>0$.  When $\mu^2=0$, the Carter constant $Q$ can be negative, which case can be considered similar to the procedure that we are doing now.

We next introduce the action variables $J_\varphi:=p_\varphi=L_z$ and \begin{equation}\label{EqJthetaK}J_\theta=\frac{1}{2\pi}\oint p_\theta\,d\theta=\frac{1}{\pi}\int_{\theta^-}^{\theta^+} \sqrt\Theta\,d\theta=\frac{1}{\pi}\int_{\theta^-}^{\theta^+} \sqrt{Q-U(\theta)}\,d\theta,\end{equation}
where $\theta^-<\theta^+$ are two roots of the integrand (c.f. \eqref{EqVerticalK}). We further introduce their dual angular variables $\al_\theta, \al_\varphi\in \T$ respectively using the symplectic form (c.f. Section 50 of \cite{A89}). In the new coordinates, we get that the Hamiltonian is a function of $J_\theta$ and $L_z$ only, and the the dynamics of $(\al_\theta, \al_\varphi)$ on $\T^2$ is a linear flow with frequency $\omega=(\omega_\theta, \omega_\varphi)=(\frac{\partial H_N}{\partial J_\theta},\frac{\partial H_N}{\partial L_z})$. Again the KAM nondegeneracy and the twist condition amount to show that the ratio $\nu:=\omega_\varphi/\omega_\theta$ has no critical point.

Here we point out some essential differences between the massive case $(\mu^2=1)$ and massless case $(\mu^2=0)$. In the case of $\mu^2=0$, we note that the $E$ variable can be eliminated by the rescaling $L_z\mapsto L_z/E,\ Q\mapsto Q/E^2,\ p_\theta\mapsto p_\theta/E$. The RHS of equation \eqref{Eqellq0} has no dependence on the particle's $E$, nor does the Hamiltonian \eqref{EqHamPhotonK}.
Thus the dynamics, in particular the frequency ratio $\nu$, on the photon shell has no dependence on $E$ so in the following it is enough to choose $E=1$. However, for the massive case, the energy $E$ does play a role in the frequency ratio $\nu$.

In the proof of the next theorem, we provide explicit formula for the ratio $\nu$. Let $a,M,E$ be fixed constants. Restricted to the energy level $H_N(J_\theta, L_z)=-\frac12\mu^2$, the ratio $\nu$ depends only on one parameter that can be chosen to be either $J_\theta$ or $L_z$. To be more physically relevant, we introduce a new parametrization as follows. 
 In the Kerr case, the photon shell is foliated by bound unstable orbits with constant radius varying in an interval $\mathcal R$.  Recall that in equation \eqref{Eqellq} we have expressed $L_z/E$ and $Q/E$ as functions of $r_c$ which is the radius of the photo sphere.  Since both $(L_z,Q)$ and $(L_z,J_\theta)$ can be used as independent integrals of the Hamiltonian $H_N$, there is a map $f$ relating the two sets of integrals: $g(Q,L_z)=(J_\theta,L_z)$. Therefore $\nu\circ g$ can be considered as a function of $r_c$ using \eqref{Eqellq}. We thus obtain a reparametrization of the ratio $\nu$ in terms of the radius $r_c$ of the slice of the photon shell.

\begin{figure}
\begin{minipage}[t]{0.5\linewidth}
\centering
\includegraphics[width=2.in]{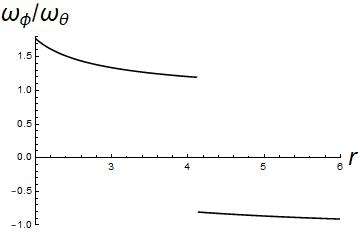}
\subcaption{$a=0.8,\ M=1,\ E=0.96$}
\label{FigPhotonK}
\end{minipage}%
\begin{minipage}[t]{0.5\linewidth}
\centering
\includegraphics[width=2.in]{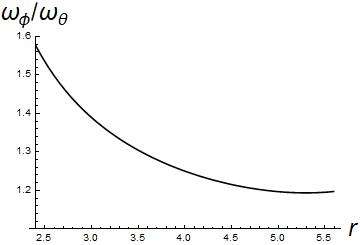}
\subcaption{$a=0.8,\ M=1,\ E=0.92$}
\label{FigPhotonK1}
\end{minipage}
\caption{The frequency ratio $\omega_\varphi/\omega_\theta$ in the $\mu^2=1$ case}
\end{figure}




\begin{Thm}\label{ThmPhotonKerr}Let $ \eps h_{\mu\nu}dx^\mu dx^\nu$ be a stationary perturbation of the Kerr metric. In either of the following two settings
\begin{enumerate}
\item In the massless case $\mu^2=0$, there exists an open set of parameter $a/M$;
\item In the massive case $\mu^2=1$,  there exists an open set of parameters $a,M,E$.
\end{enumerate}such that the following holds. Denote by $U\subset \mathcal R$ a small neighborhood (may be empty) of the point of discontinuity of $\nu\circ g(Q(r_c),L_z(r_c)).$
For $\eps$ sufficiently small, we have
\begin{enumerate}
\item There is a closed subset $\mathcal C$ of $\mathcal R\setminus U$ with measure $|\mathcal R\setminus(U\cup\mathcal C)|=O(\eps^{1/2})$ such that for each $r_c\in \mathcal C$ in the perturbed system there is a photon sphere on which all  orbit is periodic or quasiperiodic with frequency $\nu\circ g(Q(r_c),L_z(r_c))$ and all orbit lies in $\mathcal M_\nu$.
\item Each open interval in $\mathcal R\setminus (U\cup\mathcal C)$ corresponds to a Birkhoff instability region where the photon spheres are broken in the sense that $\mathcal M_\nu$ with $\nu\circ g(Q(r_c),L_z(r_c)),\ r_c\in \mathcal R\setminus \mathcal C,$  are all broken (see Definition \ref{DefBroken}), and twist map theory $($c.f. Theorem \ref{ThmAM} $(2.b), (3.b)$ and Theorem \ref{ThmMather}$)$ applies.
\end{enumerate}
\end{Thm}
\begin{proof}
Let $f$ be the map relating the two set of constants of motions $(J_\theta, L_z)$ and $(H,Q)$, i.e. $f(J_\theta, L_z)=(H,Q)$. From $Df\cdot Df^{-1}=I$, we obtain
$$\left[\begin{array}{ccc}
\frac{\partial H}{\partial J_\theta}&\frac{\partial H}{\partial L_z}\\
\frac{\partial Q}{\partial J_\theta}&\frac{\partial Q}{\partial L_z}
\end{array}\right]\left[\begin{array}{ccc}
\frac{\partial J_\theta}{\partial H}&\frac{\partial J_\theta}{\partial Q}\\
\frac{\partial L_z}{\partial H}&\frac{\partial L_z}{\partial Q}
\end{array}\right]=I.$$
This implies $Df=\frac{1}{\det Df^{-1}}\left[\begin{array}{ccc}
\frac{\partial L_z}{\partial Q}&- \frac{\partial J_\theta}{\partial Q}\\
-\frac{\partial L_z}{\partial H}&\frac{\partial J_\theta}{\partial H}
\end{array}\right].$ Denote by $$\omega=(\omega_\theta,\omega_\varphi)=\frac{1}{\det Df^{-1}}\left( \frac{\partial L_z}{\partial Q},- \frac{\partial J_\theta}{\partial Q}\right)$$ the first row of $Df$.
We next provide explicit formula for the derivatives in the above expression (noting that $U$ depends on $L_z$ and $L_z$ is a function of $Q$)
\begin{equation*}
\begin{aligned}
\frac{\partial J_\theta}{\partial Q}&=\frac{1}{2\pi}\int_{\theta^-}^{\theta^+} \frac{1}{\sqrt{Q-U(\theta)}} \,d\theta-\frac{L_z}{\pi}\frac{\partial L_z}{\partial Q}\int_{\theta^-}^{\theta^+} \frac{\cos^2\theta}{\sin^2\theta}\frac{1}{\sqrt{Q-U(\theta)}} \,d\theta.
\end{aligned}
\end{equation*}
We are only interested in the ratio $\frac{\omega_\varphi}{\omega_\theta}$, then we get
\begin{equation}\label{EqFrequencyK}
\begin{aligned}
\frac{\omega_\varphi}{\omega_\theta}&=- \frac{\partial J_\theta}{\partial Q}/\frac{\partial L_z}{\partial Q}\\
&=- \left(\frac{\partial L_z}{\partial Q}\right)^{-1}\frac{1}{2\pi}\int_{\theta^-}^{\theta^+} \frac{1}{\sqrt{Q-U(\theta)}} \,d\theta+\frac{L_z}{\pi}\int_{\theta^-}^{\theta^+} \frac{\cos^2\theta}{\sin^2\theta}\frac{d\theta}{\sqrt{Q-U(\theta)}}.\end{aligned}\end{equation}
Up to now the formula remains the same formally for the $\mu^2=0$ and $\mu^2=1$ cases.
We next express $\frac{\omega_\varphi}{\omega_\theta}$ as a function of $r_c$. It is enough to express $L_z$ and  $Q$ as a function of $r_c$ in \eqref{Eqellq} for the $\mu^2=1$ case and \eqref{Eqellq0} for the $\mu^2=0$ case.  
We get $\frac{\partial L_z}{\partial Q}=\frac{d\ell}{d r}/\frac{dq}{dr}$ from \eqref{Eqellq}  for the $\mu^2=1$ case \eqref{Eqellq0} for the $\mu^2=0$ case. Thus the ratio $\omega_\varphi/\omega_\theta$ is a function of $r_c$ in both cases.

In the massless case, the ratio $\omega_\varphi/\omega_\theta$ is independent of $E$. Figure \ref{FigPhotonK0} is plotted with choice of parameter $a/M=0.8$.
In the massive case, we plot Figure \ref{FigPhotonK} with parameter $a=0.8,M=1,E=0.96$ and Figure \ref{FigPhotonK1} with parameter $a=0.8,M=1,E=0.92 $, corresponding to Figure \ref{fig:side:aChandra} and \ref{fig:side:bChandra} respectively. We note that Figure \ref{FigPhotonK} is discontinuous which is related to the fact that $L_z$ changes sign, and we have that \ref{FigPhotonK1} is continuous since $\ell$ does not change sign in Figure \ref{fig:side:bChandra}. We shall explain the discontinuity in the next section. Also there is a critical point in \ref{FigPhotonK1}, so we can only apply KAM and twist map theory outside a small neighborhood of the critical point.
\end{proof}

\subsubsection{The point of discontinuity}\label{SSSDisc}

 There is an interesting feature in the graph in Figure \ref{FigPhotonK0} and Figure \ref{FigPhotonK} that is the appearance of a discontinuity point.  The point of discontinuity occurs at the place where $\ell(r_c)=0$ in \eqref{Eqellq}. This means that the orbit becomes from prograde to retrograde as $r$ crosses the zero of $L_z$ when increasing. This phenomenon was noticed and explained in \cite{Wi,T}. Let us give a more mathematical explanation.

 We see from Section \ref{SSSVertical} that the vertical dynamics in the $L_z=0$ case differs drastically from the $L_z\neq 0$ case. When $L_z=0$ and $\bar Q:=Q-(\mu^2-E^2)a^2>0$, the latitudinal motion $\theta$ ranges over $[0,\pi]$, in particular, the orbit passes through the south and north pole. For $|L_z|\neq 0$ but sufficient small, the latitudinal angle $\theta\in [\dt,\pi-\dt]$ where $\dt$ solves the equation $Q-U(\dt)=0$. In particular, as $L_z\to0$ we have $\frac{L_z^2}{\sin^2\dt}\to \bar Q$.

 The discontinuity originates from the second term on the RHS of \eqref{EqFrequencyK}. Another way to see it is to evaluate the quantity $\Delta\varphi$ that is the change of the azimuthal angle during each period of the $\theta$-motion. We have
 $$\Delta\varphi=\oint\frac{d\varphi}{d\theta}d\theta=2\int^{\pi-\dt}_{\dt}\frac{\frac{a}{\Delta}P(r)-\left(aE-\frac{L_z}{\sin^2\theta}\right)}{\sqrt{Q-U(\theta)}}d\theta.$$
 We shall now work on $\Delta\varphi$ and it will turn out that the the discontinuity in $\Delta\varphi$ agrees with the second term of the RHS of \ref{EqFrequencyK}.

 We note that $\frac{d\varphi}{d\lambda}= \frac{L_z}{\sin^2\theta}+O(1)=\frac{\bar Q}{L_z}+O(1)$ as $\theta\to \dt$ or $\pi-\dt$.
 Then we have $\frac{d\varphi}{d\theta}\simeq \frac{L_z}{\sin^2\theta}\frac{1}{\sqrt{Q-U}}\simeq  \frac{L_z}{\sin^2\theta\sqrt{\bar Q-\frac{L_z^2}{\sin^2\theta}}}$ when $\theta$ is $c\sqrt\dt$-close to $\dt$ or $\pi-\dt$ for large number $c$. For $\theta\in [c\sqrt\dt, \pi-c\sqrt\dt]$, the $L_z$ dependence in the integrand of $\Delta\varphi$ is negligible so the integral $\int_{c\sqrt\dt}^{\pi-c\sqrt\dt} \frac{d\varphi}{d\theta}d\theta$ is almost the same for $L_z>0$ and $L_z<0$ with small modulus. The contribution to $\Delta\varphi$ when $\theta$ varies in the interval $[\dt,c\sqrt\dt]$ is given by
 \begin{equation*}
 \begin{aligned}
& L_z\int_\dt^{c\sqrt\dt}\frac{d\theta}{\sin^2\theta \sqrt{\bar Q-\frac{L_z^2}{\sin^2\theta}}}\simeq  L_z\int_\dt^{c\sqrt\dt}\frac{d\theta}{\theta \sqrt{\theta^2\bar Q-\dt^2\bar Q}}\\
 =&\mathrm{sgn}(L_z)\int_1^{c/\sqrt\dt}\frac{dx}{x\sqrt{x^2-1}}
 \simeq\mathrm{sgn}(L_z)\int_0^1\frac{dy}{\sqrt{1-y^2}}\\
 =&\mathrm{sgn}(L_z)\frac\pi2
 \end{aligned}
 \end{equation*}
 where we substitute $L_z\simeq \sqrt{\bar Q}\dt,\ \sin\theta\simeq\theta$ and the $\simeq$ becomes $=$ when we take limit $L_z\to 0$ hence $\dt\to 0$. During each period of the $\theta$ motion, this discontinuity should be counted four times, thus the total discontinuity when $L_z$ changes sign is $4\pi$.  In particular, when $L_z$ changes sign, the orbit changes from prograde to retrograde. Similar calculation also occurs in \eqref{EqFrequencyK}. Since we use a different normalization than $\Delta\varphi$, the discontinuity in \eqref{EqFrequencyK} is 2.

 We remark that for some choices of parameters $a,E$, the function $\ell(r)$ in  \eqref{Eqellq} may be strictly positive (see Figure \ref{fig:side:bChandra}), in which case the ratio $\nu$ does not change sign.



\subsection{KAM and QPO}

In this section, we consider only the timelike case with $\mu^2=1$.

Suppose $E^2<1$ and $-R$ attains its local max at $r^*(>r_h)$ and a local min at $r_*>r^*$. Suppose also $-R(r^*)>0$ and $-R(r_*)<0$ (imagine that we shift the graph of $-R$ in Figure \ref{FigR} upward slightly).
Let $C>0$ be a large constant, and we define the following bounded part of the phase space where all Kerr orbits are bounded hence Liouville-Arnold theorem applies
\begin{equation}\label{EqBK}\mathcal B(C):=\left\{ E^2<1,\  |L_z|<C,\ 0<Q<C\ |\ -R(r^*)>0,\ -R(r_*)<0\right\}\cap H^{-1}(-\frac12). \end{equation}
This is a neighborhood of the stable circular orbit.
 Denote by $r^-<r^+$ the two largest roots of $-R(r)=0$, then we have $r^*<r_-<r_*<r_+$. 
In the $(\theta,p_\theta)$-component, the dynamics is again periodic around the point $(\theta,p_\theta)=(\pi/2,0)$ as we have analyzed in Section \ref{SSSVertical}.

 Again, similar to the Schwarzchild case, we consider for simplicity perturbations that are stationary and axisymmetric, thus $E$ and $L_z$ remain constants of motion in the perturbed system.  Therefore we fix values of $L_z$ and $E$, and introduce the action-angle coordinates $(\al_r,J_r,\al_\theta, J_\theta)\in \T\times \R_+\times \T\times \R_+$ as follows  \begin{equation*}
\begin{aligned}
J_r&=\frac{1}{2\pi}\oint p_r\,dr=\frac{1}{\pi}\int_{r^-}^{r^+}\frac{\sqrt{R(r)}}{\Delta(r)}\,dr,\\
J_\theta&=\frac{1}{2\pi}\oint p_\theta\,d\theta=\frac{1}{\pi}\int_{\theta^-}^{\theta^+}\sqrt{Q-U(\theta)}\,d\theta,
\end{aligned}
\end{equation*}
where $r^-<r^+$ and $\theta^-<\theta^+$ are zeros of the integrands above respectively. Again we get the Hamiltonian written in these coordinates $H=H(J_r,J_\theta)$ and we have the
$\omega_r:=\frac{\partial H}{\partial J_r}$ the radial frequency and $\omega_\theta:=\frac{\partial H}{\partial J_\theta} $ the latitudinal frequency. Both the expressions of $J_r$ and $J_\theta$ depend on $Q$, so we can write the frequency ratio $\omega_r/\omega_\theta$ as a function of $Q$.

\begin{Thm}\label{ThmQPOK} Let $\eps h_{\mu\nu} dx^\mu dx^\nu$ be a stationary and axisymmetric perturbation to the Kerr metric. Then there exists an open set of parameters $a,E,L_z$ such that there are a neighborhood $\mathcal D$ of the point $(r,p_r,\theta,p_\theta)=(r_*,0,\pi/2, 0)$ and $\eps_0$ such that KAM and the twist map theory are applicable to the perturbed system in $\mathcal D$ and for all $0<|\eps|<\eps_0$.
\end{Thm}

\begin{figure}
\centering
\includegraphics[width=0.4\textwidth]{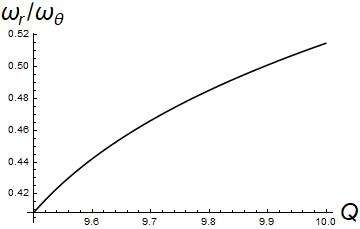}
\caption{The ratio $\omega_r/\omega_\theta$ for Kerr with parameters $a=0.8, E=0.96, L_z=1.5$}
\label{FigQPOK}
\end{figure}
Figure \ref{FigQPOK} is plotted with  parameters $a=0.8, E=0.96, L_z=1.5$.

 An example of such perturbations is given by the Monko-Novikov metric \cite{GAC}, which is so complicated that we do not cite it here. In \cite{GAC} the authors analyzed the $1:2$ and $2:3$ resonances and discovered that the dynamics of the perturbed Kerr is chaotic. The frequency ratio $\omega_r/\omega_\theta$ has been studied by many authors, c.f. \cite{BGH,S} etc.


\begin{proof}[Proof of Theorem \ref{ThmQPOK}]
Let $f$ be the map relating two sets of constants of motion $(J_r,J_\theta)$ and $H,Q$: $f(J_r, J_\theta)=(H,Q)$. From $Df\cdot Df^{-1}=I$, we obtain
$$\left[\begin{array}{ccc}
\frac{\partial H}{\partial J_r}&\frac{\partial H}{\partial J_\theta}\\
\frac{\partial Q}{\partial J_r}&\frac{\partial Q}{\partial J_\theta}
\end{array}\right]\left[\begin{array}{ccc}
\frac{\partial J_r}{\partial H}&\frac{\partial J_r}{\partial Q}\\
\frac{\partial J_\theta}{\partial H}&\frac{\partial J_\theta}{\partial Q}
\end{array}\right]=I.$$
This implies $Df=\frac{1}{\det Df^{-1}}\left[\begin{array}{ccc}
\frac{\partial J_\theta}{\partial Q}&- \frac{\partial J_r}{\partial Q}\\
-\frac{\partial J_\theta}{\partial H}&\frac{\partial J_r}{\partial H}
\end{array}\right].$ Denote by $$\omega=(\omega_r,\omega_\theta)=\frac{1}{\det Df^{-1}}\left( \frac{\partial J_\theta}{\partial Q},- \frac{\partial J_r}{\partial Q}\right)$$ the first row of $Df$.
In our setting, we fix $H=-1/2, E<1$ and $L_z$ as constants, then the only remaining parameter is $Q$. We next provide explicit formula for the derivatives in the above expression
\begin{equation}
\begin{aligned}
\frac{\partial J_r}{\partial H}&=\frac{1}{\pi}\int_{r^-}^{r^+}\frac{r^2}{\sqrt R}\,dr,\quad
\frac{\partial J_r}{\partial Q}=-\frac{1}{2\pi}\int_{r^-}^{r^+} \frac{1}{ \sqrt R} \,dr,\\
\frac{\partial J_\theta}{\partial H}&=\frac{2}{\pi}\int_{\theta^-}^{\theta^+}\frac{a^2\cos^2\theta}{\sqrt{Q-U(\theta)}}\,d\theta,\quad
\frac{\partial J_\theta}{\partial Q}=\frac{1}{2\pi}\int_{\theta^-}^{\theta^+} \frac{1}{\sqrt{Q-U(\theta)}} \,d\theta.
\end{aligned}
\end{equation}
Then we get $$\frac{\omega_r}{\omega_\theta}=-\frac{\partial J_\theta}{\partial Q}/ \frac{\partial J_r}{\partial Q}= \int_{\theta^-}^{\theta^+} \frac{1}{\sqrt{Q-U(\theta)}} \,d\theta\Big/\int_{r^-}^{r^+} \frac{1}{ \sqrt R} \,dr.$$
Figure \ref{FigQPOK} is then plotted with this formula. This completes the proof of the theorem.

\end{proof}

Finally, for the sake of completeness, we provide formula for the verification of the KAM nondgeneracy conditions in Section \ref{SSKAM}. This part is similar to Section \ref{SSSDerivatives} so readers are suggested to skip it when first reading.
We need to evaluate $\nabla \omega$ where $\nabla $ is the derivative with respect to $J_r,J_\theta$. We thus arrive at
\begin{equation*}
\begin{aligned}\nabla\omega&=\frac{\partial(H,Q)}{\partial(J_r,J_\theta)}(\partial_{H},\partial_{Q})\omega=Df \cdot (\partial_{H},\partial_{Q})\left(\frac{1}{\det Df^{-1}}\left( \frac{\partial J_\theta}{\partial Q},- \frac{\partial J_r}{\partial Q}\right)\right)\\
&=\left(\frac{1}{\det Df^{-1}}Df \right) (\partial_{H},\partial_{Q})\left( \frac{\partial J_\theta}{\partial Q},- \frac{\partial J_r}{\partial Q}\right)\\
&+Df \cdot \left((\partial_{H},\partial_{Q})\left(\frac{1}{\det Df^{-1}}\right)\right)\otimes \left( \frac{\partial J_\theta}{\partial Q},- \frac{\partial J_r}{\partial Q}\right).
\end{aligned}\end{equation*}
We can then compute  $\nabla \omega$ similar to Section \ref{SSSDerivatives}. 

If we are interested in perturbations also depending on the angle $\varphi$, then we have a Hamiltonian system of three degrees of freedom and we have to consider three fundamental frequencies as in  Section \ref{SSS3Freq}. The procedure is similar to the Schwarzschild case in Section \ref{SSS3Freq}, though the computations are more involved. We skip the details here.
\appendix

\section{The theorem of normally hyperbolic invariant manifold}\label{SNHIM}
In this section we give the version of normally hyperbolic invariant manifold theorem that we used in the proof of our main theorem. The standard references are \cite{HPS, F}. 
\begin{Def}[NHIM]\label{DefNHIM} Let $N\subset M$ be a submanifold $($maybe noncompact$)$ invariant under $f$, $f(N) = N$. We say that $N$ is a normally hyperbolic invariant manifold if there exist a constant $C > 0$, rates $0 <\lb <\mu^{-1} < 1$ and an invariant $($under $Df)$ splitting for every $x \in N$
\[T_x M = E^s(x) \oplus E^u(x) \oplus T_xN
\]
in such a way that
\begin{equation}\nonumber
\begin{aligned}
v\in E^s(x)\quad &\Leftrightarrow\quad  |Df^n(x)v| \leq C\lb^n |v|, \quad n \geq 0,\\
v\in E^u(x)\quad &\Leftrightarrow \quad |Df^n(x)v| \leq C\lb^{|n|} |v|, \ \ n \leq 0,\\
v\in T_xN\quad &\Leftrightarrow \quad |Df^n(x)v| \leq C\mu^n |v|, \quad n \in \Z.
\end{aligned}
\end{equation}
Here the Riemannian metric $|\cdot|$ can be any prescribed one, which may change the constant $C$ but not $\lambda,\mu$.
\end{Def}

\begin{Thm}\label{ThmNHIM}Suppose $N$ is a NHIM under the $C^r$, $r>1$, diffeomorphism $f:\ M\to M$. Denote $\ell=\min\{r,\frac{|\ln\lambda|}{|\ln\mu|}\}$.
Then for any $C^r$ $f_\epsilon$ that is sufficiently close to $f$ in the $C^1$ norm,
\begin{enumerate}
\item  there exists a NHIM $N_\epsilon$ that is a $C^\ell$ graph over $N$,
\item$(${\it Invariant splitting}$)$ There exists  a splitting for $x\in N_\epsilon$
\begin{equation*}\label{EqSplitting}T_x M=E^u_{\epsilon}(x)\oplus E^{s}_{\epsilon}(x)\oplus T_xN_\epsilon
\end{equation*}
invariant under the map $f_\epsilon$.  The bundle $E^{u,s}_{\epsilon}(x)$ is $C^{\ell-1}$ in $x$.
\item There exist $C^\ell$ stable and unstable manifolds $W^s(N_\epsilon)$ and $W^u(N_\epsilon)$ that are invariant under $f$ and are tangent to $E_\epsilon^s\oplus TN_\epsilon$ and $E_\epsilon^u\oplus TN_\epsilon$ respectively.
\item  The stable and unstable manifolds $W^{u,s}(N_\epsilon)$ are fibered by the corresponding stable and unstable leaves $W^{u,s}_{x,\epsilon}:=\{y\ | d(f_\eps^n(y),f_\eps^n(x))\to 0, n\to +\infty\ \mathrm{for}\ s, n\to -\infty\ \mathrm{for}\ s\}$: $$W^{u}(N_\epsilon)=\cup_{x\in N_\epsilon} W_{x,\epsilon}^u,\quad W^{s}(N_\epsilon)=\cup_{x\in N_\epsilon} W_{x,\epsilon}^s.$$ 
\item The maps $x\mapsto W_{x,\epsilon}^{u,s}$ are $C^{\ell-j}$ when $W_{x,\epsilon}^{u,s}$ is given $C^j$ topology.
\item If $f$ and $f_\epsilon$ are Hamiltonian, then $N_\epsilon$ is symplectic and the map $f_\epsilon$ restricted to $N_\epsilon$ is also Hamiltonian.
\end{enumerate}
\end{Thm}

\section{Aubry-Mather theory for twist maps}\label{AppAM}
  The dynamics of twist maps are studied in details in Aubry-Mather theory.  Let us summarize the main result and refer readers to \cite{B} for more detailed proofs. The theory can be put in a very general setting. For the sake of concreteness, we consider the following setting. Let $H:\ T^*\T^1\times \T^1\to \R$ be a nonautonomous 1-periodic Hamiltonian satisfying
 \begin{enumerate}
 \item $H(x,y,t)=H(x+1,y,t)=H(x,y,t+1)$;
 \item $\partial_y^2 H>0$ everywhere;
 \item $\frac{H(x,y,t)}{|y|}\to\infty$ as $|y|\to\infty$ for all $(x,t)\in \T^2$;
 \item The flow $\phi^t$ is complete, i.e. for each initial condition $z_0=(x_0,y_0)$, the orbit $\phi^t(z_0)$ is defined for all $t\in \R$.
 \end{enumerate}
 The time-1 map, $\phi^1:\ T^*\T^1\to T^*\T^1$ of  the  Hamiltonian flow is symplectic and is the composition of twist maps \cite{M1}.

 Let $H:\ T^*\T^1\times \T^1\to \R$ be as above and we introduce its dual Lagrangian by taking Legendre transform $L(x,\dot x,t):=\sup_y( \dot x\cdot y-H(x,y,t))$, and for any curve $\gamma:\ [0,n]\to \T^1,\ n\in \Z$, we define its action
 $$A(\gamma)=\int_0^nL(\gamma(t),\dot\gamma(t),t)\,dt.$$
 A curve is a critical point of the action with fixed endpoints if and only if it solves the Euler-Lagrange equation $\frac{d}{dt}\frac{\partial L}{\partial \dot x}(\gamma(t),\dot\gamma(t),t)=\frac{\partial L}{\partial x}(\gamma(t),\dot\gamma(t),t)$. If the curve $\gamma:\ \R\to\T^1$ solves the Euler-Lagrange equation, then the Legendre transform ($\{(\gamma(t),y(\gamma(t)))\}\subset T^*\T^1$) of ($\{\gamma(t),\dot\gamma(t))\}\subset T\T^1$) solves the Hamiltonian equation, then the time-1 map $\phi^1$ maps $(\gamma(n),y(\gamma(n)))$ to $(\gamma(n+1),y(\gamma(n+1)))$ for all $n\in \Z$.
We say that a $C^1$ curve $\gamma:\ \R\to \T^1$ is globally minimizing, if for any $r,s\in \R$ and any $\xi:\ [r',s']\to \T^1$ with $r'-r\in \Z, s'-s\in \Z$, we have
$$\int_r^sL(\gamma(t),\dot\gamma(t),t)\,dt\leq \int_{r'}^{s'}L(\xi(t),\dot\xi(t),t)\,dt.$$
Since the Lagrangian $L$ is $1$-periodic in its first entry, we may lift each curve $\gamma:\ \R\to \T^1$ to a curve $\tilde \gamma:\ \R\to \R$ by unfolding $\T^1$ to its universal covering space $\R$, so we can visualize the graph of $\tilde \gamma$ in $\R^2$.

The following theorem summarizes the main results of the twist map theory (c.f. \cite{B}).
 \begin{Thm}\label{ThmAM}
 \begin{enumerate}
 \item Each globally minimizing curve $\gamma$ has a well-defined rotation number $\rho=\lim_{n\to\pm\infty}\frac{\tilde \gamma(n)}{n}$. For each $\rho$, the set of globally minimizing curves with rotation number $\rho$ is not empty, we introduce the set \begin{equation*}
 \begin{aligned}
 \mathcal M_\rho&=\{(\gamma(n),\dot\gamma(n))\in T\T^1,\ n\in \Z\ |\\
 &\ \gamma:\ \R\to \T^1\ \mathrm{ is\ globally\ minimizing\ with \ rotation\ number}\ \rho\}.
 \end{aligned}
 \end{equation*}
 We also use the same notation to denote its Legendre transform to $T^*\T^1$.
  \item If the rotation number $\rho$ is irrational, $\mathcal{M}_\rho$ is
  \begin{enumerate}
 \item  either an invariant circle (a homotopically nontrivial circle on $T^*\T^1$),
\item  or a Denjoy minimal set (a Cantor subset of a homotopically nontrivial circle on $T^*\T^1$).
\end{enumerate}
  \item If the rotation number $\rho$ is rational, $\mathcal{M}_\rho$ always contains periodic orbits.
  \begin{enumerate}
  \item It may be an invariant circle consists of periodic orbits of the same period.
  \item If $\cM_\rho$ is not an invariant circle, then in the gap between two neighboring periodic orbits $\tilde \gamma^-,\, \tilde \gamma^+$ in $\R^2$ where $\{(\gamma^\pm(n),\dot \gamma^\pm(n))\}\subset \mathcal M_\rho$, there is a globally minimizing curve $\gamma$ with rotation number $\rho$ whose lift is asymptotic to $\tilde \gamma^-$ in the future and to $\tilde \gamma^+$ in the past, and another  asymptotic to $\tilde \gamma^+$ in the future and to $\tilde \gamma^-$ in the past.
  \end{enumerate}
\end{enumerate}

 \end{Thm}

\begin{Def}\label{DefBroken}
 When $\mathcal M_\rho$ is not a homotopically nontrivial invariant circle in $T^*\T^1$, we say that $\mathcal M_\rho$ is \emph{broken}. This includes case (2.b) and (3.b) in the above Theorem \ref{ThmAM}.
 \end{Def}
Let $R$ be the region (called Birkhoff region of instability) bounded by $\phi^1$-invariant homotopically non-trivial Jordan curves $\Gamma_-$ and $\Gamma_+$ with $\rho(\Gamma_-) < \rho(\Gamma_+)$, where $\rho(\Gamma_\pm)$ is the rotation number of $\phi^1|_{\Gamma_\pm}$.  Note that all $\mathcal M_\rho\subset R$ with $\rho(\Gamma_-)<\rho<\rho(\Gamma_+)$ is broken. We cite the following theorem of Mather.
 \begin{Thm}[\cite{Ma2}]\label{ThmMather}
\begin{enumerate}
\item Suppose $\rho(\Gamma_-) < \alpha, \omega < \rho(\Gamma_+)$. Then there is an orbit of $\phi^1$ whose $\alpha$-limit set lies in $\mathcal M_\alpha$ and whose $\omega$-limit set lies in $\cM_\omega$. Furthermore, if $\rho( \Gamma_- )$ $($resp. $\rho(\Gamma_+))$ is irrational, then this conclusion still holds with the weaker hypothesis $\rho(\Gamma_-) \leq \al$, (resp. $\al,\ \omega\leq \rho(\Gamma_+))$.
\item
Consider for each $i\in \Z$ a real number $\rho(\Gamma_-)\leq\omega_i\leq \rho(\Gamma_+)$ and a positive number $\eps_i$. Then there exists an orbit $(\ldots, P_j,\ldots)$ and an increasing bi-infinite sequence of integers $j(i)$ such that dist. $(P_{j(i)},),\cM_{\omega(i)})<\eps_i$.
\end{enumerate}
\end{Thm}

\section*{Acknowledgment}
The author would like to express his deep gratitudes to Professor S.-T. Yau for suggesting the topic. He also would like to thank Mr. Yifan Guo for drawing all the figures. The author is supported by NSFC (Significant project No.11790273) in China and Beijing Natural Science Foundation (Z180003).


\begin{thebibliography}{DGO}

\def\bi#1{\bibitem[#1]{#1}}
\bi{A89} Arnol'd, V.I. {\it Mathematical methods of classical mechanics}. Vol. 60. Springer, 1989.
\bi{AKN}Arnold, Vladimir I., Valery V. Kozlov, and Anatoly I. Neishtadt. {\it Mathematical aspects of classical and celestial mechanics.} Vol. 3. Springer Science \& Business Media, 2007.

\bi{B} Bangert, Victor. {\it Mather sets for twist maps and geodesics on tori.} Dynamics reported. Vieweg+ Teubner Verlag, 1988. 1-56.
\bi{BGH} Brink, Jeandrew, Marisa Geyer, and Tanja Hinderer. {\it The Astrophysics of Resonant Orbits in the Kerr Metric.} arXiv preprint arXiv:1501.07728 (2015).
\bi{C} Chandrasekhar, Subrahmanyan. {\it The mathematical theory of black holes.} Research supported by NSF. Oxford/New York, Clarendon Press/Oxford University Press (International Series of Monographs on Physics. Volume 69), 1983, 663 p. 1 (1983).
\bi{Ca} Dias Carneiro, {\it On minimizing measures of the action of autonomous Lagrangians}, Nonlinearily, 8 (1995)1077-1085.
\bi{DKN} Efthimia Deligianni, Jutta Kunz, Petya Nedkova, {\it Quasi-periodic oscillations from the accretion disk around distorted black holes}, Arxiv: 2003.01252.
\bi{F} Fenichel N., {\it Persistence and smoothness of invariant manifolds for flows}, Indiana Univ. Math. J {\bf 21} (1971) 193-226.
\bi{GAC}Lukes-Gerakopoulos, Georgios, Theocharis A. Apostolatos, and George Contopoulos. {\it Observable signature of a background deviating from the Kerr metric.} Physical Review D 81.12 (2010): 124005.
\bi{GHW} Gralla, Samuel E., Daniel E. Holz, and Robert M. Wald. {\it Black hole shadows, photon shells, and lensing rings.} Physical Review D 100.2 (2019): 024018.
\bi{GL} Gralla, Samuel E., and Alexandru Lupsasca. {\it Null geodesics of the Kerr exterior.} Physical Review D 101.4 (2020): 044032.
\bi{GLP} Grossman, Rebecca, Janna Levin, and Gabe Perez-Giz. {\it Harmonic structure of generic Kerr orbits.} Physical Review D 85.2 (2012): 023012
\bi{HF} Hinderer, Tanja, and Eanna E. Flanagan. {\it Two-timescale analysis of extreme mass ratio inspirals in Kerr spacetime: Orbital motion.} Physical Review D 78.6 (2008): 064028.
\bi{HPS} Hirsch M. Pugh C. \& Shub M., {\it Invariant Manifolds}, Lecture Notes Math. {\bf 583} (1977) Springer-Verlag.
\bi{IM} Adam R. Ingram, Sara E. Motta, {\it A review of quasi-periodic oscillations from black hole X-ray binaries: observation and theory}, New Astronomy Reviews 85 (2019) 101524
\bi{J} Michael D. Johnson, Alexandru Lupsasca, Andrew Strominger, et al, {\it Universal interferometric signatures of a black hole's photon shell.} arXiv:1907.04329v2 [astro-ph.IM] 27 Mar 2020
\bi{JDW} Chichuan Jin, Chris Done, and Martin Ward, {\it Reobserving the NLS1 galaxy RE J1034+396 - I. The long-term, recurrent X-ray QPO with a high significance}, MNRAS 495, 3538-3550 (2020).
\bi{M1} Moser, J,  {\it Monotone twist mappings and calculus of variations}, Erg. Theory Dyn. Sys, 1986, 6, 401-413.
\bi{M2} Moser, J. {\it On invariant curves of area-preserving mappings of an annulus.} Nachr. Akad. Wiss. G\"ottingen, II (1962): 1-20.
\bi{M87} Wielgus, M, {\it et al},  {\it Monitoring the Morphology of M87* in 2009-2017 with the Event Horizon Telescope},  Astrophysics Journal, vol 901, no 67, 2020.
\bi{Ma1} Mather J., {\it Action minimizing invariant measures for positive definite Lagrangian systems,} Math. Z., {\bf 207(2)}(1991) 169-207.
\bi{Ma2}Mather, John N. {\it Variational construction of orbits of twist diffeomorphisms.} Journal of the American Mathematical Society (1991): 207-263.
\bi{N} New York Times, {\it Infinite Visions Were Hiding in the First Black Hole Image's Rings}, https://www.nytimes.com/2020/03/28/science/black-hole-rings.html.
\bi{P}  P\"{o}schel, J\"urgen. {\it A lecture on the classical KAM theorem.} arXiv preprint arXiv:0908.2234 (2009).
\bi{S}Schmidt, Wolfram. {\it Celestial mechanics in Kerr spacetime.} Classical and Quantum Gravity 19.10 (2002): 2743.
\bi{T} Teo, Edward. {\it Spherical photon orbits around a Kerr black hole.} General Relativity and Gravitation 35.11 (2003): 1909-1926.
\bi{WBS} Warburton, Niels, Leor Barack, and Norichika Sago. {\it Isofrequency pairing of geodesic orbits in Kerr geometry.} Physical Review D 87.8 (2013): 084012.
\bi{Wi} D. Wilkins, {\it Bounded geodesics in the Kerr metric}, Physical Review D, Vol 5, No. 4, 814-822, 1972,
\bi{X} J. Xue, {\it Arnold diffusion and geodesic dynamics of blackholes}, preprint.
\end{thebibliography}
\end{document}